\documentclass[11pt]{amsart}
\usepackage[latin1]{inputenc}
\usepackage[T1]{fontenc}
\usepackage{hyperref}
\usepackage{amsfonts}
\usepackage{amsmath}
\usepackage{epsf, subfigure, verbatim}
\usepackage{graphicx}
\usepackage{amssymb}
\usepackage{amsthm}
\usepackage{latexsym}
\usepackage{layout}
\usepackage[comma, authoryear]{natbib}
\usepackage{color}

\def\sn{\sqrt{n}}

\newcommand{\reels}{\mathbb{R}}

\newcommand{\esp}{\mathbb{E}}
\newcommand{\p}{\mathbb{P}}

\newcommand{\procX}{(X(t))_{t \geq 0}}

\newcommand{\mean}{\mu}

\newcommand{\Mb}{M_{b}}
\newcommand{\Msigma}{M_{\tilde{\sigma}}}
\newcommand{\lambdaud}{\lambda_{12}}
\newcommand{\lambdadu}{\lambda_{21}}
\newcommand{\lambdauu}{\lambda_{11}}
\newcommand{\lambdadd}{\lambda_{22}}
\newcommand{\Em}{E}
\newcommand{\tE}{\mathbb{R}_{+}\times\Em}

\newtheorem{thm}{Theorem}
\newtheorem{lem}[thm]{Lemma}

\newtheorem{prop}[thm]{Proposition}
\newtheorem{RM}[thm]{Remark}
\newtheorem{ex}[thm]{Example}
\newtheorem{assump}[thm]{Condition}

\begin{document}

\title[A simple discretization scheme]{A simple discretization scheme for nonnegative diffusion processes, with applications to option pricing}
\maketitle

\vspace{1cm}

\begin{center}
{\sc Chantal Labb\'e}\footnote{corresponding author}\\
Service de l'enseignement des m\'ethodes quantitatives de gestion\\
HEC Montr\'eal\\
3000 chemin de la C\^{o}te-Sainte-Catherine\\
Montr\'eal, Qu\'ebec, Canada\\
H3T 2A7\\
Phone number: 514-340-6723\\
E-mail: chantal.labbe@hec.ca\\

\vspace{1cm}
{\sc Bruno R\'emillard}\\
Service de l'enseignement des m\'ethodes quantitatives de gestion\\
HEC Montr\'eal\\
3000 chemin de la C\^{o}te-Sainte-Catherine\\
Montr\'eal, Qu\'ebec, Canada\\
H3T 2A7\\
Phone number: 514-340-6794\\
E-mail: bruno.remillard@hec.ca\\

\vspace{1cm}
{\sc Jean-Fran\c{c}ois Renaud}\\
D\'epartement de math\'ematiques\\
Universit\'e du Qu\'ebec \`a Montr\'eal (UQAM)\\
201 av.\ Pr\'esident-Kennedy\\
Montr\'eal, Qu\'ebec, Canada\\
H2X 3Y7\\
Phone number: 514-987-3000 ext.\ 7675\\
E-mail: renaud.jf@uqam.ca\\

\vspace{1cm}
\today
\end{center}

\thispagestyle{empty}
\newpage

\begin{center}
\begin{minipage}{14cm}
{\sc Abstract}. A discretization scheme for nonnegative diffusion processes is proposed and the convergence of the corresponding sequence of approximate processes is proved using the martingale problem framework. Motivations for this scheme come typically from finance, especially for path-dependent option pricing. The scheme is simple: one only needs to find a nonnegative distribution whose mean and variance satisfy a simple condition to apply it. Then, for virtually any (path-dependent) payoff, Monte Carlo option prices obtained from this scheme will converge to the theoretical price. Examples of models and diffusion processes for which the scheme applies are provided.
\end{minipage}
\end{center}

\vspace{0.5cm}



\section{Introduction}

The Cox-Ingersoll-Ross (CIR) process, also known as the mean-reverting square-root diffusion, was introduced by \cite{coxetal1985} for interest rates modeling. It now has other financial applications, for example in Heston's stochastic volatility model \citep{heston1993}, where it plays the role of the squared volatility. This process is the solution to the following stochastic differential equation:
\begin{equation}\label{E:cir}
dX(t) = \kappa (\beta - X(t)) \, dt \, + \nu \sqrt{X(t)} \, dW(t) , \quad X(0) = x_0 ,
\end{equation}
where $(W(t))_{t \ge 0}$ is a one-dimensional standard Brownian motion, $\kappa$, $\beta$ and $\nu$  are strictly positive constants, and the initial value satisfies $x_0\geq0$. It is known that this process stays nonnegative.

If one wants to construct a discrete-time approximation of~$(X(t))_{t\geq0}$ defined by~\eqref{E:cir}, 
a standard Euler-Maruyama scheme cannot be applied directly. Indeed, for a time-step of size $1/n$, the approximating process $(Y_n(k))_{k\geq0}$ would be given by $Y_n(0)=x_0$, and
\begin{equation}\label{E:CIR schema d'Euler}
Y_n(k+1) = Y_n(k) + \frac{\kappa}{n}(\beta - Y_n(k)) + \frac{\nu}{\sn} \sqrt{Y_n(k)}\; Z_k^{(n)},
\end{equation}
for $k=0,1,2,\dots$, where the $Z_k^{(n)}$'s are random variables with the standard normal distribution. 
If $Y_n(k)$ is nonnegative for a given~$k$, then~\eqref{E:CIR schema d'Euler} correctly defines $Y_n(k+1)$, whose value then has a non-zero probability of being negative. If $Y_n(k+1)$ actually happens to be negative, the next iteration cannot be achieved since it would involve taking its square root. Thus, one can iterate~\eqref{E:CIR schema d'Euler} only until the first $k$ for which $Y_n(k)<0$.
%
Several ways have been proposed for avoiding this problem. For example, one may simulate values of $Z_k^{(n)}$ until the right-hand side of~\eqref{E:CIR schema d'Euler} is nonnegative, and then set $Y_n(k+1)$ to be this value, but this results in a scheme for which the number of steps needed to generate a sample of a given size is random.
One could also use either of these schemes instead of~\eqref{E:CIR schema d'Euler}:
\begin{itemize}
\item[(b1)] $Y_n(k+1) = Y_n(k) + \frac{\kappa}{n}(\beta - Y_n(k)) + \frac{\nu}{\sn} \sqrt{(Y_n(k))^+} \;Z_k^{(n)}$;
\item[(b2)] $Y_n(k+1) = Y_n(k) + \frac{\kappa}{n}(\beta - (Y_n(k))^+) + \frac{\nu}{\sn} \sqrt{(Y_n(k))^+} \;Z_k^{(n)}$;
\item[(b3)] $Y_n(k+1) = \bigr|\; Y_n(k) + \frac{\kappa}{n}(\beta - Y_n(k)) + \frac{\nu}{\sn} \sqrt{Y_n(k)} \;Z_k^{(n)}\; \bigr|$;
\item[(b4)] $Y_n(k+1) = Y_n(k) + \frac{\kappa}{n}(\beta - Y_n(k)) + \frac{\nu}{\sn} \sqrt{\,|Y_n(k)|\,} \;Z_k^{(n)}$,
\end{itemize}
where $x^+:=\max(x,0)$, $x\in\reels$.
All these schemes are well defined, but (b1), (b2) and (b4) generate processes $(Y_n(k))_{k\geq0}$ whose values are not necessarily nonnegative. In quantitative finance, this is often perceived as a drawback when approximating positive quantities, such as interest rates, stock prices and volatilities. For more details on Euler-Maruyama discretization schemes for diffusions, the reader is referred to \cite{glasserman2004} and \cite{kloedenplaten1992}.

A natural yet crucial question when using a discretization scheme is: does it converge to the process we wish to approximate as the time-step decreases to zero? Classical theory mostly deals with diffusions with Lipschitz coefficients, excluding the CIR process (whose diffusion coefficient is not Lipschitz).  \cite{deelstradelbaen1998} establish the strong convergence of the scheme (b1), in a framework where the mean reversion parameter~$\beta$ may be a stochastic process. More recently, \cite{bossydiop2007} and \cite{berkaouietal2008} studied the weak and strong 
 convergence of the scheme (b3), under the more general setting where the diffusion coefficient in~\eqref{E:cir} is replaced with $\nu (X(t))^\alpha$, for some $\alpha\in[1/2,1)$. Note that by letting $\alpha=1/2$, one retrieves the CIR process. \cite{highammao2005} study 
strong convergence in the case of (b4). \cite{lordetal2010} introduce (b2), a modification of (b1), discuss its strong convergence, and present an overview of several discretization schemes, including (b1)-(b4) and implicit ones, and present numerical comparisons. Similarly, \cite{alfonsi2005} presents implicit schemes (that admit analytical solutions), studies their weak and strong convergence, and presents numerical comparisons with (b1) and (b3). Among popular implicit methods, let us mention the implicit Milstein scheme, described for instance in \cite{kahletal2008}, where it is found to be better for discretizing the CIR process than the explicit Milstein scheme or the balanced implicit method of \cite{milsteinetal1998}.

In financial engineering, one is often interested in pricing a derivative security for which the CIR process is involved in modeling the underlying asset. However, it seems like very little attention has been given to the question of convergence of approximate prices (resulting from the discretization) to the right price, especially for path-dependent derivatives.
In the present paper, we address this problem in the even more general framework where the price of the underlying asset follows a (typically nonnegative) diffusion process with time-dependent coefficients $dX(t) = b(t,X(t)) \, dt + \sigma(t,X(t)) \, dW(t)$.
The discounted payoff function of, for instance, an option is a function of the {\it path} of the underlying asset price, say $g(X)$. The price of this option is therefore given by $\esp \left[ g(X) \right]$, when the underlying probability measure is a risk-neutral measure. The goal is now to define a sequence of approximating processes $(X_n)_{n \geq 1}$, based on a discretization scheme, such that $\esp \left[ g(X_n) \right]$ converges to $\esp\left[ g(X) \right]$ as $n$ goes to infinity. Indeed, Monte Carlo estimators of $\esp \left[ g(X_n) \right]$ will then provide numerical values of the right price $\esp\left[ g(X) \right]$. Thus, existence of a weak solution $X$ to the above SDE, and convergence in distribution of a sequence of processes $(X_n)_{n \geq 1}$ to this solution is usually more than sufficient for pricing purposes. Indeed, such convergence is equivalent to $\esp \left[ g(X_n) \right]$ converging to $\esp\left[ g(X) \right]$ 
for all path functionals~$g$ within a class of sufficiently well behaved functionals. 
The discounted payoff function $g$ of an option is typically a {\it continuous} (or almost surely continuous) function of the path of the underlying asset, and this is usually sufficient to use the previous definition, 
establishing at once price convergence.
It is to be stressed that convergence in distribution of the sequence of processes $(X_n)_{n \geq 1}$ to $X$ involves the distribution of the whole path and must not be confused with convergence in distribution of the sequence of random variables $(X_n(T))_{n\geq1}$ to $X(T)$, where $T$ is the maturity time. The latter is a much weaker statement and is useful for pricing European contingent claims for which the payoff depends only on the value of the underlying security at the maturity date, but generally does not allow us to deal with path-dependent contingent claims. Weak convergence results generally found in the literature (such as in \cite{bossydiop2007} or \cite{alfonsi2005}) are of this latter type, i.e., they pertain to the processes sampled at a fixed instant $T>0$. Note that the convergence in distribution of the sequence of processes $(X_n)_{n \geq 1}$ to $X$ also includes the convergence in distribution of any sequence of random vectors $(X_n(t_1),\dots,X_n(t_k))$ to $(X(t_1),\dots,X(t_k))$ for fixed times $t_1,\dots,t_k$.

As pointed out earlier, much attention is dedicated to strong convergence in the literature. Strong convergence roughly says that the approximating process is uniformly close to~$X$ on the interval of time $[0,T]$ for large~$n$, hence it holds promises to establish price convergence for path-dependent contingent claims with maturity~$T$. \cite{highammao2005} actually take this next step when $X$ is the CIR process; they define continuous-time approximation processes from the discrete scheme (b4), prove their strong convergence towards~$X$, and then deduce convergence of the price of a few path-dependent derivative securities. Note that many papers deal with the CIR process in isolation; see, e.g., \cite{alfonsi2005}, \cite{berkaouietal2008}, and \cite{bossydiop2007}. As opposed to that, \cite{highammao2005} prove convergence of the price of a barrier option in Heston's model, in which the CIR process is used to model the squared volatility of the stock price. To the best of our knowledge, they are the first to establish, by showing convergence for certain option prices, that using an Euler-type discretization in the full Heston model is theoretically correct. Numerical results, obtained from several discretization schemes, are provided in~\cite{lordetal2010} for some options in the Heston model.
Note that the transition density function of the CIR process is known to be (within a scaling parameter) noncentral chi-square, which allows for direct simulation of this process; \cite{broadiekaya2006} have provided an exact simulation algorithm for Heston's model. However, algorithms using the transition density are computationally slower than Euler-type schemes, especially when the trajectory must be sampled at a large number of time points. Therefore, this family of algorithms is less suited for pricing highly path-dependent options; see for example the introductory discussion in~\cite{highammao2005}. For this reason, and the fact that direct discretization methods are widely used in practice, our focus is on the latter.

So, \cite{highammao2005} show that strong convergence of the approximating processes to~$X$ may effectively be used to deduce price convergence for  certain contingent claims. Such a deduction involves somewhat delicate calculus of probability, the complexity of which depends on the complexity of the specific contingent claim considered.
As previously discussed, one could instead easily deduce price convergence by essentially verifying that the payoff function is continuous, provided that one had a sequence of approximation processes $(X_n)_{n \geq 1}$ weakly converging towards~$X$, in the sense of convergence in distribution of the whole path. It is possible to implement this other approach, with relative ease and in general setups, by using a powerful idea set forth by Stroock and Varadhan (discussed in detail in~\cite{stroockvaradhan1979}), namely the characterization of Markov processes by means of the so-called {\it martingale problem}. In particular, Stroock and Varadhan's approach provides an alternative way to regard diffusions. This point of view has the advantage of being particularly well suited to establish convergence of Markov chains to diffusion processes. Actually, a martingale problem is entirely defined by the expression of a {\it generator}, and it turns out that it is sufficient to establish convergence of the generators of a sequence of Markov chains to the {\it infinitesimal generator} of a diffusion, for this sequence of Markov chains to converge in distribution to the diffusion.

In this paper, our main goal is to propose a discretization scheme, in the form of a Markov chain with nonnegative values, and use the above-mentioned techniques based on the generator to show its convergence in distribution to the solution of the SDE, under suitable assumptions on~$b$ and~$\sigma$. Note that in a standard Euler-Maruyama scheme, such as the one presented in~\eqref{E:CIR schema d'Euler} for the CIR process, it is the use of the normal distribution which is responsible for the non-zero probability of getting a negative value on the next time-step, even if the current value is nonnegative. When we work within the framework of the martingale problem, there is no additional difficulty in establishing convergence if we trade the normal distribution for another distribution whose second moment is finite, in the spirit of the \textit{weak Euler scheme} (see, e.g., \cite{kloedenplaten1992}). This suggests the following idea: let us use a scheme very similar to the standard Euler-Maruyama scheme, where by a careful choice of distribution we make sure that the resulting Markov chain is well defined and assumes only nonnegative values. We propose to use a nonnegative distribution, and we give conditions on its mean and variance to ensure that the scheme is well defined and converges. Thus, application of our scheme in practice reduces to the sole choice of a nonnegative distribution whose mean and variance satisfy a given condition, making it relatively simple and versatile. The family of diffusions for which such a choice of distribution is possible encompasses several examples of practical interest. For instance, our scheme applies to the CIR process, including some cases where the reversion parameter is a (possibly correlated) stochastic process (as in~\cite{deelstradelbaen1998}), and it can be used in the framework of Heston's model. Moreover, verifying that the approximate prices of a path-dependent option converge to the real price is then a matter of verifying that the payoff functional is continuous (at least on a set of probability~1), and hence one does not have to resort to delicate calculus of probability; note that it is left for future work to analytically evaluate the rate of weak convergence of this scheme. Finally, it is interesting to note that the convergence of the binomial tree towards the geometric Brownian motion, and that of the GARCH(1,1) discrete process to the continuous version of this process, are special cases of our scheme.

The paper is organised as follows. In Section~\ref{sec:schema}, we introduce the nonnegativity preserving discretization scheme, and specify for which diffusions it applies. In Section~\ref{sec:convergence}, we present the main result (Theorem~\ref{T:main}), which says that the approximate processes defined by the scheme converge in distribution to the right diffusion process. A brief introduction to the martingale problem is provided. In Section~\ref{sec:implications}, we discuss the consequences of the convergence in distribution established in the main result. In particular, we see the ease with which it can be used in the pricing of derivative securities. Some examples of dynamics for which the scheme applies (and converges) are given in Section~\ref{sec:examples}. Finally, the results of two numerical experiments involving the proposed scheme as well as schemes (b1)-(b4) are presented in Section~\ref{sec:numerical}.

\section{The scheme}\label{sec:schema}

Consider once again the stochastic differential equation (SDE)
\begin{equation}\label{E:mainsde}
dX(t) = b(t,X(t)) \, dt + \sigma(t,X(t)) \, dW(t) ,\; t\geq0,\quad X(0)=x_0,
\end{equation}
where the following is satisfied:
\begin{assump}\label{A:cond0}
The coefficients $b:\reels_+\times\reels^d\rightarrow\reels^d$ and $\sigma:\reels_+
\times\reels^d\rightarrow\reels^d\otimes\reels^d$ are continuous functions. Moreover, $x_0$ is a constant in~$\reels^d$. Here, $d$ is a positive integer, $\reels_+:=[0,\infty)$, and $\reels^d\otimes\reels^d$ is the space of $d\times d$-matrices.
\end{assump}
The process $(W(t))_{t\geq0}$ represents a $d$-dimensional standard Brownian motion. The main reason to allow equation~\eqref{E:mainsde} to be multidimensional is to set a framework general enough to include as special cases, for instance, two-factor interest rate models or stochastic volatility models, where at least two SDEs are involved simultaneously. We are primarily interested in cases where the solution to~\eqref{E:mainsde} is a $\reels^d$-valued process, at least one of whose components is a nonnegative process. More precisely, we postulate the following:
\begin{assump}\label{A:cond1}
There is an integer~$m$, with $0\leq m\leq d$, such that for each $x_0\in\Em$, where $\Em:=\reels_+^m\times\reels^{d-m}$, the SDE~\eqref{E:mainsde} has a unique (in the sense of probability law) weak solution, that is there exists a probability space $(\Omega,\mathcal{F},\p)$, a filtration $\{\mathcal{F}_t\}_{t\geq0}$ satisfying the usual conditions, and a pair $(W,X)$, where $W$ is a $d$-dimensional standard Brownian motion and $X$ is a process with continuous paths satisfying~\eqref{E:mainsde}. Moreover, the unique weak solution is such that $X(t)\in\Em$ for all $t\geq0$, almost surely.
\end{assump}

\begin{RM}
Note that only the values of the  functions $b$ and $\sigma$ over $\tE$ are relevant under this setting.
\end{RM}

In order to fix ideas, let us first focus on the case where $x_0\geq0$ is given, and $\procX$ is a nonnegative one-dimensional process (i.e., $d=m=1$ and $E=\reels_+$). Moreover, assume that $\sigma(\cdot)$ is a nonnegative (scalar) function. If we fix a discrete time-step of size $1/n$, for some integer~$n$, then a discrete-time approximating process $(Y_n(k))_{k\geq0}$ is defined as follows: let $Y_n(0)=x_0$ and, for each $k\geq0$, set
\begin{equation}\label{eq:schema une dim}
Y_n(k+1)=Y_n(k)+\frac1n b(k/n,Y_n(k))+\frac{1}{\sn}\sigma(k/n,Y_n(k))(\varepsilon_{k}^{(n)}-\mean),
\end{equation}
where $(\varepsilon_k^{(n)})_{k\geq0}$ is a family of independent copies of a random variable~$\varepsilon$ with mean~$\mean$ and variance~1. In a standard Euler scheme, $\varepsilon-\mean$ follows the standard normal distribution, causing $Y_n(k+1)$, even given $Y_n(k)\geq0$, to have a non-zero probability of being negative. To avoid this problem, let us instead assume that $\varepsilon$ is a nonnegative random variable. If its (positive) mean~$\mean$ is set such that
\begin{equation}\label{eq:condition pour positivite, 1 dim}
x+\frac1n b(t,x) - \frac1{\sn}\sigma(t,x)\,\mean\geq0,\quad\text{for all }(t,x)\in\reels_+\times\reels_+,
\end{equation}
then $(Y_n(k))_{k\geq0}$ is clearly a Markov chain whose values are nonnegative. We will see later that letting go of the normality is not too much of a price to pay, as $(Y_n(k))_{k\geq0}$ provides a valid approximation (see Theorem~\ref{T:main}). We also delay the illustration of the above simple scheme, in the case of the CIR process (in Example~\ref{ex:exemple du CIR une dim}), until after its generalization to the multidimensional case, that we now undertake.

Assume the general setting of Conditions~\ref{A:cond0} and~\ref{A:cond1}. Put $a:=\sigma\sigma^\top$ ($\top$ indicates the transpose operation). The multidimensional discretization scheme relies on a function $\tilde{\sigma}$, a sequence of independent copies of a random vector~$\varepsilon$ on a probability space $(\tilde{\Omega},\tilde{\mathcal{F}},\tilde{\p})$, and an integer~${n_0}$, chosen to satisfy:
\begin{assump}\label{A:cond2}
We have $a(t,x)=\tilde{\sigma}(t,x)\Sigma\,\tilde{\sigma}^\top(t,x)$ for all $(t,x)\in\tE$, where~$\Sigma$ is a symmetric semi-definite positive $d\times d$ matrix, and $\tilde{\sigma}:\tE\rightarrow\reels^d\otimes\reels^d$ is a continuous function. The positive integer ${n_0}$ and the mean $\mean$ are such that
\begin{equation}\label{eq:condition de nonnegativite}
\inf_{(t,x)\in\tE} \left( x + \frac{1}{n} b\left(t,x\right) - \frac{1}{\sn} \tilde{\sigma}\left(t,x\right)\mean \right) \in\Em,\quad\text{for all }n\geq {n_0},
\end{equation}
where the infimum is taken componentwise. Moreover, the random vector~$\varepsilon$ has mean $\mu$, covariance matrix~$\Sigma$, and
\begin{equation}\label{eq:nonnegativite des epsilon}
\tilde{\p}(\tilde{\sigma}(t,x)\varepsilon\in\Em\text{ for all }(t,x)\in\tE)=1\,.
\end{equation}
\end{assump}
\begin{RM}
Equation~\eqref{eq:condition de nonnegativite} is a direct extension of condition~\eqref{eq:condition pour positivite, 1 dim}. The condition in equation~\eqref{eq:nonnegativite des epsilon} is satisfied, for instance, if all components of $\tilde{\sigma}$ are nonnegative functions, and all components of $\varepsilon$ are nonnegative random variables. When $d=1$ and $\Em=\reels_+$, and in the typical situation where the scalar function~$\sigma$ is nonnegative, we usually use $\tilde{\sigma}=\sigma$ and $\Sigma=1$. In this case, as mentioned previously, upon setting a nonnegative distribution for~$\varepsilon$ with variance $\Sigma=1$ and mean~$\mean>0$ satisfying~\eqref{eq:condition de nonnegativite} (or equivalently~\eqref{eq:condition pour positivite, 1 dim}), then  Condition~\ref{A:cond2} is satisfied.
\end{RM}
\begin{RM}
Note that $\mean$ and ${n_0}$ are typically chosen as a couple in order to satisfy the condition given in equation~\eqref{eq:condition de nonnegativite}, providing more flexibility in the choice of the distribution of~$\varepsilon$.
\end{RM}
Let $(\varepsilon_k^{(n)})_{k\geq0,\,n\geq {n_0}}$ be a family of independent copies of the random vector~$\varepsilon$ of Condition~\ref{A:cond2} on the probability space $(\tilde{\Omega},\tilde{\mathcal{F}},\tilde{\p})$. Fix some $x_0\in\Em$. For $n\geq {n_0}$, define $(Y_n(k))_{k\geq0}$ by letting $Y_n(0)=x_0$, and then by iterating as follows:
\begin{equation}\label{E:algo}
Y_n(k+1) = Y_n(k) + \frac{1}{n} b\left(k/n, Y_n(k)\right) + \frac{1}{\sn} \tilde{\sigma} \left(k/n, Y_n(k) \right) (\varepsilon_{k}^{(n)} - \mean),
\quad k\in\{0,1,2,\dots\}.
\end{equation}
For each $n\geq {n_0}$, the process $(Y_n(k))_{k\geq0}$ is a Markov chain with values in~$\Em$, i.e., $\tilde{\p}(Y_n(k)\in\Em)=1$ for all $k\geq0$, in view of Condition~\ref{A:cond2} and the fact that $x_0\in\Em$. Finally, define the time-continuous approximating process $(X_n(t))_{t\geq0}$ by linear interpolation of the values of $(Y_n(k))_{k\geq0}$ between the discrete time-steps:
\begin{equation}\label{eq:linear interpolation}
X_n(t) = Y_n (\lfloor nt \rfloor) + (nt - \lfloor nt \rfloor)(Y_n (\lfloor nt \rfloor+1)-Y_n (\lfloor nt \rfloor)),\quad t\geq0,
\end{equation}
where $\lfloor y \rfloor$ is the largest integer less than or equal to $y$. By construction, $(X_n(t))_{t\geq0}$ is an $\Em$-valued process for each $n\geq {n_0}$.
Our main result, discussed in Section~\ref{sec:convergence}, states that the sequence of approximating processes $(X_n)_{n\geq {n_0}}$ converges to the (unique) solution of the SDE~\eqref{E:mainsde}. Section~\ref{sec:implications} discusses further the nature of this convergence and its applications in option pricing.

We conclude this section by giving an example of dynamics for which Conditions~\ref{A:cond0}, \ref{A:cond1}, and~\ref{A:cond2} are satisfied, and hence the scheme is applicable. Section~\ref{sec:examples} gives more such examples of practical interest.

\begin{ex}\label{ex:exemple du CIR une dim}{\bf Cox-Ingersoll-Ross (CIR) model} 
The CIR process is defined by equation~\eqref{E:mainsde} with $b(t,x) = \kappa (\beta - x)$ and $\sigma(t,x) = \nu \sqrt{x}$, $x\geq0$, where $\kappa$, $\beta$ and $\nu$ are positive constants. Here, $d=m=1$ and $\Em=\reels_+$, and we trivially set $\tilde{\sigma}=\sigma$ and $\Sigma=1$.  With straightforward optimization computations, one can show that whenever $n > \kappa$ and $\mean>0$ we have
\begin{equation}\label{eq:inf pour le CIR}
\inf_{x \ge 0} \left( x + \frac{1}{n} \kappa (\beta - x) - \frac{\mean}{\sn} \nu \sqrt{x} \right) = \frac{\kappa \beta}{n} - \frac{\mean^2 \nu^2}{4(n-\kappa)}.
\end{equation}
Consequently, if we choose ${n_0}>\kappa$, and $\mean$ such that the right-hand side of~\eqref{eq:inf pour le CIR} is nonnegative, which is true if $0<\mean\leq\frac2\nu\sqrt{\kappa\beta\left(1-\frac{\kappa}{{n_0}}\right)}$, and finally set a nonnegative distribution with mean~$\mean$ and variance~1 for~$\varepsilon$, then Condition~\ref{A:cond2} is satisfied.
It is also interesting to note that~\eqref{E:algo} implies $\tilde{\esp}[Y_n(k+1)]=\kappa\beta/n+\tilde{\esp}[Y_n(k)](1-\kappa/n)$, $k=0,1,2,\dots$, from which we deduce $\tilde{\esp}[Y_n(k)]=\beta+(x_0-\beta)(1-\kappa/n)^k$. Using~\eqref{eq:linear interpolation}, one may easily conclude that, for a fixed time $T>0$, $\lim_{n\to\infty}\tilde{\esp}[X_n(T)]=\beta+(x_0-\beta)e^{-\kappa T}=\esp[X(T)]$, where $(\Omega,\mathcal{F},\p)$, $\{\mathcal{F}_t\}_{t\geq0}$, $(W,X)$ is the weak solution of~\eqref{E:mainsde}; see \citep[Section 4.4]{shreve2004b} for the last equality. Very similarly, one may establish that the variance of $X_n(T)$ goes to that of $X(T)$ as $n$ goes to infinity.
\end{ex}

\section{Convergence of the scheme using the martingale problem formulation}\label{sec:convergence}

In this section, we establish convergence of the probability law of the processes $X_n$, $n\geq {n_0}$ (defined in the previous section), towards the law of the solution to the SDE~\eqref{E:mainsde}, as $n$ goes to infinity (or the time-step goes to zero). We achieve this by means of the martingale problem of Stroock and Varadhan. We provide a very brief introduction to the martingale problem in this section, mainly establishing the notation, and refer any reader seeking for a detailed discussion to \cite{ethierkurtz1986} or \cite{stroockvaradhan1979}.

Let $C_c^\infty(\reels^d)$ be the set of infinitely differentiable functions $f:\reels^d\rightarrow\reels$ with compact support. Define the differential operator
$A$, acting on functions $f\in C_c^\infty(\reels^d)$, by
\begin{equation}\label{E:gen}
(A f)(t,x) := \sum_{i=1}^d b_i(t,x) \partial_{x_i} f(x) + \frac{1}{2} \sum_{i=1}^d \sum_{j=1}^d a_{ij}(t,x) \partial_{x_i} \partial_{x_j} f(x),
\end{equation}
$(t,x)\in\reels_+\times\reels^d$, where $\partial_{x_i}$ stands for the partial derivative with respect to the $i$-th variable, and $b_i$ and $a_{ij}$ denote the entries of~$b$ and~$a$.
If $(\Omega,\mathcal{F},\p)$, $\{\mathcal{F}_t\}_{t\geq0}$, $(W,X)$ is a weak solution to the SDE~\eqref{E:mainsde}, then  $\p(X(0)=x_0)=1$, and from It\^o's formula one easily deduces that
\begin{equation}\label{eq:martingale}
f(X(t))-\int_0^t A f(s,X(s)) ds
\end{equation}
is an $\{\mathcal{F}_t\}$-martingale for any $f\in C_c^\infty(\reels^d)$.
Actually, any process $X$ with continuous paths, defined on a probability space $(\Omega,\mathcal{F},\p)$ endowed with a filtration $\{\mathcal{F}_t\}_{t\geq0}$, satisfying $\p(X(0)=x_0)=1$ and such that~\eqref{eq:martingale} is an $\{\mathcal{F}_t\}$-martingale for all $f\in C_c^\infty(\reels^d)$, is called a solution to the martingale problem for $(A,x_0)$.
Hence, any weak solution to the SDE~\eqref{E:mainsde} solves the martingale problem for $(A,x_0)$. Under Condition~\ref{A:cond0}, the converse turns out to be true. In fact, if the martingale problem for $(A,x_0)$ has a solution, then there exists a weak solution to the SDE~\eqref{E:mainsde} (see Corollary 5.3.4 in \cite{ethierkurtz1986}).
Moreover, uniqueness (in the sense of probability law) holds for solutions of the SDE~\eqref{E:mainsde} if and only if uniqueness holds for the martingale problem for $(A,x_0)$ (see again Corollary 5.3.4 in \cite{ethierkurtz1986}).
In other words, existence and uniqueness of a weak solution to the SDE~\eqref{E:mainsde} is equivalent to existence and uniqueness of a solution to the martingale problem for $(A,x_0)$. This means that all the information which is truly essential relatively to the SDE~\eqref{E:mainsde} is encapsulated in the martingale problem formulation, which in turn relies only on the expression of~$A$, also called the infinitesimal generator.

In view of the importance of the martingale problem and the generator, let us expose similar ideas in discrete time, that is, let us see how the Markov chain $(Y_n(k))_{k\geq0}$ defined in~\eqref{E:algo} may be characterized through a martingale problem as well. Consider the transition function
\begin{equation}\label{eq:transition function}
K_n(t,x,\Gamma):= \tilde{\p}\left( x+\frac1n b(t,x) + \frac1{\sn}\tilde{\sigma}(t,x)(\varepsilon-\mean)\in\Gamma\right),\; \Gamma\in\mathcal{B}(\reels^d),
\end{equation}
for $(t,x)\in\reels_+\times\Em$, where $\mathcal{B}(\reels^d)$ represents the Borel sets on $\reels^d$. This function actually describes the transitions of the Markov chain $(Y_n(k))_{k\geq0}$. Indeed, let $\{\mathcal{F}_k^{Y_n}\}_{k\geq0}$ be the filtration generated by $(Y_n(k))_{k\geq0}$ (i.e. $\mathcal{F}_k^{Y_n}:=\sigma\{Y_n(l)|l=0,1,\dots,k\}$), and note that
\begin{equation}\label{eq:transitions de Ynk}
\tilde{\p}(Y_n(k+1)\in\Gamma\,|\,\mathcal{F}_k^{Y_n})=K_n(k/n,Y_n(k),\Gamma),\quad\Gamma\in\mathcal{B}(\reels^d),
\end{equation}
for $k\in\{0,1,\dots\}$.
For all $f\in C_c^\infty(\reels^d)$, define
\begin{equation}\label{eq:generateur discret}
A_n f(t,x) := n\int(f(y)-f(x))K_n(t,x,dy),
\end{equation}
$(t,x)\in\reels_+\times\Em$. For each $f\in C_c^\infty(\reels^d)$, it is easily seen that 
\begin{equation}\label{eq:martingale temps discret}
f(Y_n(k))-\frac1n\sum_{l=0}^{k-1} A_nf(l/n,Y_n(l)),
\end{equation}
is a $\{\mathcal{F}_k^{Y_n}\}$-martingale. Conversely, any discrete-time process $(\tilde{Y}_n(k))_{k\geq0}$ such that $\tilde{Y}_n(k)=x_0$ and \eqref{eq:martingale temps discret} (with $\tilde{Y}_n$ in place of~$Y_n$) is a martingale for all $f\in C_c^\infty(\reels^d)$ (with respect to the filtration generated by $(\tilde{Y}_n(k))_{k\geq0}$) is a Markov chain whose transitions are described by~$K_n$
\citep[Section 11.2]{stroockvaradhan1979}. Hence, the martingale problem defined by~\eqref{eq:generateur discret} and~\eqref{eq:martingale temps discret} characterizes such Markov chains, and $A_n$ in~\eqref{eq:generateur discret} is the discrete analogue of the infinitesimal generator~$A$.

In order to show convergence of the sequence of processes~$(X_n)_{n\geq {n_0}}$ (recall~\eqref{E:algo} and~\eqref{eq:linear interpolation}) to the weak solution of the SDE~\eqref{E:mainsde} as $n$ goes to infinity, it turns out that it is sufficient to show convergence of the sequence of generators, more precisely that $A_nf$ converges to $Af$ as $n$ goes to infinity, for all $f\in C_c^\infty(\reels^d)$. It is this method that we use to establish our main result:
\begin{thm}\label{T:main}
Under Conditions~\ref{A:cond0}, \ref{A:cond1} and~\ref{A:cond2}, the sequence of approximating processes~$(X_n)_{n\geq {n_0}}$ defined by~\eqref{E:algo} and~\eqref{eq:linear interpolation} converges in distribution to the weak solution of SDE~\eqref{E:mainsde} or, equivalently, to the solution of the martingale problem for $(A,x_0)$.
\end{thm}
The details of the proof are found in the appendix. We insist on the fact that the convergence in distribution established in Theorem~\ref{T:main} is that of the whole path of the processes~$X_n$, $n\geq {n_0}$, to the whole path of the solution to the SDE~\eqref{E:mainsde}, although we delay until the next section a careful definition of this convergence.

\section{Consequences of the convergence in distribution}\label{sec:implications}

Let us take the example where the SDE~\eqref{E:mainsde} models the evolution of stock prices in the risk-neutral world. Let $(\Omega,\mathcal{F},\p)$, $\{\mathcal{F}_t\}_{t\geq0}$, $(W,X)$ be the weak solution of~\eqref{E:mainsde}. Suppose that the discounted payoff of a derivative security is expressed as a function, say $h(X)$, of the underlying price process~$X$; note that path-dependent options are embedded in this setup. On one hand, it is well known that $\esp[h(X)]$ is the price of this derivative. On the other hand, saying, as in Theorem~\ref{T:main}, that the sequence of approximating processes $(X_n)_{n\geq {n_0}}$ (defined on $(\tilde{\Omega},\tilde{\mathcal{F}},\tilde{\p})$) converges in distribution to $X$ means by definition that
\begin{equation}\label{eq:def conv en dist}
\lim_{n\rightarrow\infty}\tilde{\esp}[g(X_n)]=\esp[g(X)],
\end{equation}
for all bounded {\it continuous} functions~$g$. Thus, if $h$ is nice enough (this is true in particular if $h$ is itself bounded and {\it continuous}), then~\eqref{eq:def conv en dist} holds with $g=h$, and the approximate price given by the scheme, i.e., $\tilde{\esp}[h(X_n)]$, is close to the actual theoretical price when $n$ is large.

For the definition of the convergence of $(X_n)_{n\geq {n_0}}$ to~$X$ to be complete and transparent, this section starts by clarifying what it means for the above-mentioned functionals~$g$ to be continuous. Then, it provides sufficient conditions for $h$ to be nice enough for the sequence of approximate prices to converge to the right price. Moreover, this section includes several examples of such nice functions, and illustrates how they can be used for pricing.

The solution of the SDE~\eqref{E:mainsde} and approximating processes $X_n$ constructed in Section~\ref{sec:schema} are all processes whose paths 
are continuous functions. In other words, these paths belong to $C_E(\reels_+)$, the set of continuous functions $x:\reels_+\rightarrow E$. Let us endow $C_E(\reels_+)$ with the topology of uniform convergence on compact subsets of $\reels_+$, induced by the metric
\begin{equation*}
\Lambda(x,y)\;:=\;\int_0^\infty e^{-u}\,\sup_{0\leq t\leq u}\left(|x(t)-y(t)|\wedge1\right)\,du\,,
\end{equation*}
for $x,y\in C_E(\reels_+)$ (where $|\cdot|$ is the Euclidean norm on~$\reels^d$ and $u\wedge v$ is the minimum of $u,v\in\reels$).
The processes $X_n$, $n\geq {n_0}$, and~$X$ may be regarded as random variables with values in the set $C_E(\reels_+)$, and convergence in distribution of $X_n$ to~$X$ means that~\eqref{eq:def conv en dist} is satisfied for all functions $g:C_E(\reels_+)\rightarrow\reels$ which are bounded, and continuous with respect to the metric~$\Lambda$. %
\begin{RM}\label{rem:continuite par rapport a la sup-norme}
Consider a positive constant~$T$ and a function $g:C_E(\reels_+)\rightarrow\reels$. If $g(x)$ depends only on the values $x(t)$, $0\leq t\leq T$, for any $x\in C_E(\reels_+)$, then $g$ is continuous with respect to the metric~$\Lambda$ if and only if it is continuous with respect to the usual and more tractable sup-metric defined as
$$
  \Lambda_T(x,y):=\sup_{0\leq t\leq T}|x(t)-y(t)|,
$$
for $x,y\in C_E(\reels_+)$.
In mathematical finance, when we consider a finite investment horizon $[0,T]$, the discounted payoff of a derivative security is usually of the form $g(X)$, for some function $g:C_E(\reels_+)\rightarrow\reels$ whose values depend only on the trajectories restricted to the interval $[0,T]$, and hence one simply has to verify continuity of~$g$ with respect to~$\Lambda_T$.
\end{RM}

We can enlarge the set of functions~$g$ such that~\eqref{eq:def conv en dist} holds, extending at once the set of admissible payoff functions. Indeed, as a direct consequence of Theorem~\ref{T:main}, together with Theorems 1.5.1 and 1.5.4 in~\cite{billingsley1968}, we get:
\begin{prop}\label{prop:conditions generales pour la convergence des esperances}
Assume Conditions~\ref{A:cond0}, \ref{A:cond1} and~\ref{A:cond2}. Let $(X_n)_{n\geq {n_0}}$ be the sequence of processes defined by~\eqref{E:algo} and~\eqref{eq:linear interpolation} on the probability space $(\tilde{\Omega},\tilde{\mathcal{F}},\tilde{\p})$, and let $(\Omega,\mathcal{F},\p)$, $\{\mathcal{F}_t\}_{t\geq0}$, $(W,X)$ be the weak solution to the SDE~\eqref{E:mainsde}. Suppose that $g:C_E(\reels_+)\rightarrow\reels$ is  continuous with respect to~$\Lambda$ except on a subset $\mathcal{D}_g\subset C_E(\reels_+)$ satisfying $\p(X\in\mathcal{D}_g)=0$. If the sequence of random variables $(g(X_n))_{n\geq {n_0}}$ is uniformly integrable, then $\tilde{\esp}[g(X_n)]\rightarrow\esp[g(X)]$ as $n\rightarrow\infty$. In particular, this is true when $g$ is bounded.
\end{prop}

\begin{RM}\label{rem:simulations Monte Carlo}
For a given $n$, the central limit theorem (or the strong law of large numbers) implies that the average of the payoff $g(X_n)$ of $N$ Monte Carlo simulations of the trajectories of $X_n$ converges to $\tilde{\esp}[g(X_n)]$. The latter is in turn very close to $\esp[g(X)]$ when $n$ is large in view of Proposition~\ref{prop:conditions generales pour la convergence des esperances}. So, the price of a derivative security is obtained in practice by letting both $n$ and $N$ tend to infinity. This is very similar to the binomial tree method used in the Black-Scholes model (whose connection with the proposed method will actually be discussed further in Section~\ref{ex:GMB time-varying coefficients}).
\end{RM}

For a positive constant~$T$ (typically the investment time horizon), a few examples of continuous payoff functionals $g:C_E(\reels_+)\rightarrow\reels$ are given by $g(x)=x_1(T)$, $g(x)=\left(\int_0^T x_1(t)dt\right)/T$, and $g(x)=\max_{0\leq t\leq T}x_1(t)$, where $x_1$ is the first component of~$x$. Their continuity with respect to~$\Lambda_T$ is easily verified. These functions are useful when dealing respectively with plain vanilla European, Asian, and lookback options. Let us further consider the functions $g(x)=\mathbb{I}\{\max_{0\leq t\leq T} x_1(t)\leq B\}$,  where $B$ is a constant and $\mathbb{I}$ denotes the indicator function, which is continuous except on the set $\mathcal{D}_g=\{x\in C_{\Em}(\reels_+):\,\max_{0\leq t\leq T}x_1(t)= B\}$, and $g(x)=\mathbb{I}\{x_1(T)\leq B\}$, continuous except on $\mathcal{D}_g=\{x\in C_{\Em}(\reels_+):x_1(T)=B\}$. These are typical of situations where barrier and binary options are involved. Finally,
note that the mapping $g:C_{\Em}(\reels_+)\rightarrow C_{\reels}(\reels_+)$ defined by $g(x)=e^{x_1}$ for $x\in C_{\Em}(\reels_+)$ is continuous. This last mapping may be useful when the log-price of an asset is modeled by the SDE~\eqref{E:mainsde}.
Combining functions such as the above allows one to express a wide variety of bounded payoffs, and then easily conclude by Proposition~\ref{prop:conditions generales pour la convergence des esperances} that the sequence of approximate prices for the corresponding derivative security converges to the right price as the time-step goes to zero. Example~\ref{ex:un exemple de payoff borne} illustrates this technique.

\begin{ex}\label{ex:un exemple de payoff borne}
Consider the function $g:C_{\Em}(\reels_+)\rightarrow \reels$ defined by
$$
g(x)\;=\;\exp\left( -\int_0^T x_1(s)\,ds\right)\times\left( K - \frac1T\int_0^T e^{x_2(s)}\,ds\right)^+\,,
$$
for $x \in C_{\Em}(\reels_+)$, where $K$ and $T$ are positive constants and $x_i$ is the i-th component of~$x$. It is continuous over $C_{\Em}(\reels_+)$, and bounded if $m\geq1$ (i.e. $x_1(t)\geq0$ for all $t\geq0$). Thus, it satisfies the assumptions of Proposition~\ref{prop:conditions generales pour la convergence des esperances}. If the first two components of a multi-dimensional process $X$, namely $X_1$ and $X_2$, correspond respectively to the short-rate and the log-price of an asset, then $g(X)$ is the discounted payoff of an Asian put option with strike price~$K$ on this asset.
\end{ex}
We can similarly tackle all the examples considered in~\cite{highammao2005} since these are for path-dependent options with bounded payoffs, for instance up-and-out call options. Unbounded payoffs can be treated via the put-call parity principle or by truncation, as in, respectively, Examples~\ref{ex:put-call parity} and~\ref{ex:truncation} below.
\begin{ex}\label{ex:put-call parity}
Assume Conditions~\ref{A:cond0}, \ref{A:cond1} and~\ref{A:cond2}, with $m\geq2$. Suppose that~\eqref{E:mainsde} models the risk-neutral evolution of~$X$, whose first component $X_1$ is the (nonnegative) short-rate and second component $X_2$ is an asset price.  Let $D(t)=\exp(-\int_0^tX_1(s)ds)$, $t\geq0$, be the discount factor process. In the postulated risk-neutral world, $\{D(t)X_2(t)\}_{t\geq0}$ is a martingale. Upon combining this with equality $X_2(T)-K=(X_2(T)-K)^+-(K-X_2(T))^+$, we get
\begin{equation}\label{eq:put-call parity plain vanilla}
\esp[D(T)\,(X_2(T)-K)^+] = X_2(0)-K\,\esp[D(T)] \;+\; \esp[D(T)\,(K-X_2(T))^+],
\end{equation}
where $T>0$ is the investment horizon and $K>0$ the strike price. On the right-hand side, both payoffs are bounded and continuous, hence we can apply Proposition~\ref{prop:conditions generales pour la convergence des esperances} to approximate the prices of the bond and the put option. The price of the call option is then deduced from the put-call parity~\eqref{eq:put-call parity plain vanilla}.
Furthermore, if $\p(\max_{0\leq t\leq T}X_2(t)=B)=0$, where $B>0$, then, for example, one can similarly price an up-and-in call barrier option from the call option price evaluated in~\eqref{eq:put-call parity plain vanilla} and the price of an up-and-out call option whose payoff is bounded:
\begin{multline*}
\esp\left[D(T)\,(X_2(T)-K)^+\,\mathbb{I}\left\{\max_{0\leq t\leq T}X_2(t)>B\right\}\right] =
\esp\big[D(T)\,(X_2(T)-K)^+\big] \\
- \esp\left[D(T)\,(X_2(T)-K)^+\,\mathbb{I}\left\{\max_{0\leq t\leq T}X_2(t)\leq B\right\}\right].
\end{multline*}
\end{ex}
\begin{ex}\label{ex:truncation}
For simplicity, assume that $d=1$. For $K > 0$, let $h(z) = (z-K)^+$, $z \in \reels$, and let $(h_k)_{k \geq 1}$ be an increasing sequence of nonnegative, continuous and bounded functions on~$\reels$ converging pointwise to $h$; for example, one can take $h_k(z) = (z-K)^+\wedge k$, 
$k=1,2,\dots$. As $\pi_T(x)=x(T)$ is a continuous function over $C_{\Em}(\reels_+)$, then $g=h \circ \pi_T$ and $g_k=h_k \circ \pi_T$ are also continuous functions over $C_{\Em}(\reels_+)$, and in the latter case it is also a bounded function. Then, by Proposition~\ref{prop:conditions generales pour la convergence des esperances}, for each $k$ we have
$$
\lim_{n \to \infty} \tilde{\esp}[h_k(X_n(T))] = \esp[h_k(X(T))] .
$$
Finally, by monotone convergence, $\esp[h_k(X(T))]$ tends to $\esp[h(X(T))]$ as $k$ goes to infinity. Consequently, for $k$ and $n$ large, $\tilde{\esp}[h_k(X_n(T))]$ is arbitrarily close to the price of the call option with maturity~$T$ and strike price~$K$.
\end{ex}

\begin{RM}
In order to establish convergence of the sequence of approximate prices to the right price for a derivative security (with a bounded payoff function), an alternative strategy is to first show strong convergence of the sequence of (continuous-time) approximate processes to the solution of the SDE~\eqref{E:mainsde}, then find a way to deduce, from this strong convergence result, price convergence. This is the approach used by \cite{highammao2005} when \eqref{E:mainsde} corresponds to the CIR dynamics. As seen in \cite{highammao2005}, the second step (the one in which price convergence is deduced) involves probabilistic arguments, which are specific to the given derivative security, and whose complexity vary depending on that of the given derivative security. One source of difficulty comes from the way the discrete-time approximation is completed in between the discrete time steps to get a continuous-time process strongly converging to the CIR process. Indeed, this completion relies on a Brownian trajectory (see (16) in \cite{highammao2005}, and strong convergence results Theorems~3.1 and~3.2 and Corollaries~3.1 and~3.2). Such trajectory is known only in theory, while practical computation of the Monte Carlo approximate price requires a trajectory known in practice. As a result, the trajectory used for the strong convergence results is different from that used in the expression of derivative's Monte Carlo price (compare (16) with (18) and, for instance, (26) or (28) in \cite{highammao2005}). Hence, the analysis necessarily requires establishing that these two trajectories are {\it close} (see Lemma~3.2 in  \cite{highammao2005}). In contrast, in the weak convergence approach proposed here, the continuous-time approximate processes whose convergence to the solution of~\eqref{E:mainsde} is shown is known in practice: it is obtained by simple linear interpolation from the discrete scheme (recall~\eqref{eq:linear interpolation}). More importantly, upon availability of such a weakly convergent process, one gets price convergence essentially by verifying that the payoff function is a continuous function of the path. As we have shown, this is quite simple and technical probabilistic arguments typical of the strong convergence approach are conveniently avoided.
\end{RM}

\section{Examples of models and diffusions}\label{sec:examples}

We now give several examples of models and diffusions for which Theorem~\ref{T:main} applies and hence the scheme converges.

\subsection{Black-Scholes model with time-dependent coefficients}\label{ex:GMB time-varying coefficients}
Let $\beta(t)$ and $\nu(t)$, $t\geq0$, be continuous functions, with $\nu(t)\geq0$ for all $t\geq0$, and suppose that $x_0>0$ is given. Assume that $\sup_{t\geq0}\nu(t)\in(0,\infty)$, and $\inf_{t\geq0}\beta(t)>-\infty$.
Upon letting $b(t,x)=x\beta(t)$ and $\sigma(t,x)=x\nu(t)$, then the solution to the stochastic differential equation~\eqref{E:mainsde} is the one-dimensional geometric Brownian motion (here, $d=m=1$).
Set $\tilde{\sigma}=\sigma$, and $\Sigma=1$. Then Condition~\ref{A:cond2} is satisfied if $\varepsilon$ is a nonnegative random variable, and ${n_0}$ and $\mu>0$ are chosen such that
$$
   x\left( 1+\frac{\beta(t)}{n}-\frac{\nu(t)}{\sn}\mu\right)\;\geq\;0\quad\text{for all }t\geq0,\,x\geq0,\,n\geq {n_0}\,.
$$
This will be the case if, for instance, ${n_0}>-2\inf_{t\geq0}\beta(t)$ and $0<\mu<\sqrt{{n_0}}/(2\sup_{t\geq0}\nu(t))$. 
For each $n\geq {n_0}$, equation~\eqref{E:algo} gives the following scheme to approximate the geometric Brownian motion with time-dependent coefficients and with initial value~$x_0$: set $Y_{n}(0)=x_0$ and then, for $k\geq0$, set
\begin{equation}\label{eq:schema pour le GMB}
Y_{n}(k+1)\;=\;Y_{n}(k)\;\left( 1+\frac{\beta(k/n)}{n}+\frac{\nu(k/n)}{\sn}(\varepsilon_k^{(n)}-\mu)\right) \,.
\end{equation}
One possible choice of distribution for $\varepsilon$ is
\begin{equation}\label{eq:bernoulli pour le MBG}
 \mathbb{P}(\varepsilon=0)\;=\;\frac{1}{1+\mu^2}\quad\text{and}
\quad \p\left(\varepsilon=\mu+\frac1{\mu}\right)\;=\;\frac{\mu^2}{1+\mu^2}\,,
\end{equation}
or, in other words, $\mu\varepsilon/(\mu^2+1)\,\sim\,\text{Bernoulli}(\mu^2/(1+\mu^2))$. One can verify that we have indeed $\esp[\varepsilon]=\mu$ and $\text{Var}(\varepsilon)=1\,(=\Sigma)$. If the coefficients are constant, i.e. $\beta(t)=\beta_0\in\mathbb{R}$, and $\nu(t)=\nu_0>0$, for all $t\geq0$, and $\varepsilon$ follows the distribution in~\eqref{eq:bernoulli pour le MBG}, then it is interesting to note that the scheme~\eqref{eq:schema pour le GMB} reduces to a recombining binomial tree, in which $Y_{n}(k+1)$ is equal to either $u_nY_{n}(k)$ or $d_nY_{n}(k)$, where $u_n=1+\frac{\beta_0}{n}+\frac{\nu_0}{\sn}\frac1{\mu}$ and $d_n=1+\frac{\beta_0}{n}-\frac{\nu_0}{\sn}\mu$. By construction, we have 
$0<d_n<1+\frac{\beta_0}n<u_n$, and it is easy to verify that
\begin{equation}\label{eq:prob of going up; binomial tree}
\p(Y_{n}(k+1)=u_nY_{n}(k))\;=\;
\frac{\mu^2}{1+\mu^2}
\;=\;\frac{(1+\frac{\beta_0}n)-d_n}{u_n-d_n}\,.
\end{equation}
If $\beta_0$ stands for the constant interest rate, and $\{Y_{n}(k)\}_{k\geq0}$ models the risk-neutral evolution of an asset price, then the probability of an up-move (equivalent to multiplying the current price by~$u_n$) found in~\eqref{eq:prob of going up; binomial tree} is consistent with the risk-neutral probability usually postulated in binomial models (see for instance \cite{shreve2004a}).

\subsection{Constant elasticity of variance (CEV) model}\label{ex:CEV}

It is possible to apply Theorem~\ref{T:main} to one dimensional diffusion processes as those in \cite{berkaouietal2008} which have a non-Lipschitz diffusion coefficient. Let $b(t,x)=b(x)$ be a Lipschitz continuous function such that $b(0) > 0$ and let $\sigma(t,x) = \nu x^{\alpha}$, with $\nu > 0$ and $\alpha \in [1/2,1)$. Under these assumptions, and with $x_0\geq0$, there exists a nonnegative strong solution to equation~\eqref{E:mainsde} \citep{berkaouietal2008}, and so $d=m=1$ and $\Em=\reels_+$. Let $\varepsilon$ be a nonnegative random variable with mean~$\mean$ and variance $\Sigma=1$, and take $\tilde{\sigma} = \sigma$. It remains to show that ${n_0}$ and $\mu>0$ can be chosen such that the condition in equation~\eqref{eq:condition de nonnegativite} is satisfied, namely
\begin{equation}\label{eq:nonnegativite CEV}
\frac{b(x)}{n}+c_n(x)\geq0,
\end{equation}
for all $x\geq0$ and $n\geq {n_0}$, where $c_n(x) := x - \frac{\nu x^{\alpha}}{\sqrt{n}} \mu$. If $b$ has Lipschitz constant $K$ (i.e. $|b(x)-b(y)|\leq K|x-y|$ for all $x,y\in\reels$), then $b(x) \geq -Kx$, for all $x \geq 0$. Thus, \eqref{eq:nonnegativite CEV} is immediate if $-Kx/n+c_n(x)\geq0$, which is equivalent to
$$
x \geq x^{(n)} := \left( \frac{\nu \mu}{\sqrt{n}} (1-K/n)^{-1} \right)^{1/(1-\alpha)},
$$
when $1-K/n\geq1-K/{n_0}>0$. Now, it remains to secure inequality~\eqref{eq:nonnegativite CEV} for $x\in[0,x^{(n)})$.

Since $x^{(n)} \to 0$ when $n\to\infty$, and since $b$ is continuous, then for any $\mu$ we can set ${n_0}$ large enough to ensure both ${n_0}>K$ and
\begin{equation}\label{eq:min b0 CEV}
 \min_{x\in[0,x^{(n)}]}b(x)\geq\frac{b(0)}{2},\quad n\geq {n_0}.
\end{equation}
We have $c_n^{\prime}(x) = 1 - \frac{\nu \alpha}{\sqrt{n} x^{1-\alpha}} \mu$ and $c_n^{\prime \prime}(x) = \frac{\nu \alpha (1-\alpha)}{\sqrt{n} x^{2-\alpha}} \mu$. Therefore, the $c_n$'s are convex functions with a global minimum in
$$
x_n := \left( \frac{\alpha \nu \mu}{\sqrt{n}} \right)^{1/(1-\alpha)} \leq x^{(n)}
$$
with value
\begin{equation}\label{eq:non-Lip case, definition of cn}
c_n(x_n) = \left( \frac{\nu \mu}{\sqrt{n}} \right)^{1/(1-\alpha)} \left( \alpha^{1/(1-\alpha)} - \alpha^{\alpha/(1-\alpha)} \right) \leq 0,
\end{equation}
where the last inequality follows from $\alpha<1$.
When $\alpha\in(1/2,1)$, we have $|nc_n(x_n)|\to0$ as $n\to\infty$, and for any $\mu$ one may choose ${n_0}$ sufficiently large to get $|nc_n(x_n)|<b(0)/2$ for all $n\geq {n_0}$. Together with~\eqref{eq:min b0 CEV}, this gives
\begin{equation}\label{eq:CEV suite d'inegalites}
\frac{b(x)}{n} + c_n(x) \geq \frac{1}{n} \min_{x \in [0,x^{(n)}]} b(x) + c_n(x_n)
\geq \frac{1}{n} \left( \frac{b(0)}{2} + n c_n(x_n) \right)\geq 0\,,
\end{equation}
for all $x\in[0,x^{(n)}]$ and $n\geq {n_0}$.
When $\alpha=1/2$, then $c_n(x_n)=-\nu^2\mu^2/(4n)$ (recall~\eqref{eq:non-Lip case, definition of cn}). Then the last inequality in~\eqref{eq:CEV suite d'inegalites} is immediate if $\mu$ satisfies $\mu<\sqrt{2b(0)}/\nu$.

\subsection{One-dimensional affine diffusion process}

Let $h_0,h_1,k_0,k_1,r_0$ be constant numbers satisfying $k_0h_1-k_1h_0>0$, $h_1\not=0$, and $h_0+h_1r_0\geq0$. Then, the process defined by
$$
 dR(t)=(k_0+k_1 R(t))\,dt+\sqrt{h_0+h_1R(t)}\,dW(t),\quad R(0)=r_0,
$$
belongs to the family of affine diffusion processes (see~\cite{duffiepansingleton2000}) and takes its values in $[-h_0/h_1,\infty)$ if $h_1>0$, or in $(-\infty,-h_0/h_1]$ if $h_1<0$. The process $(X(t))_{t\geq0}$ defined by the change of variable $X(t)=h_0+h_1R(t)$ satisfies the dynamics~\eqref{E:mainsde} with $b(t,x)=b(x)=(k_0h_1-k_1h_0) + k_1x$, $\sigma(t,x)=|h_1|\sqrt{x}$, and $x_0=h_0+h_1r_0$. This is the special case of Example~\ref{ex:CEV} with $\alpha=1/2$, $b(0)=k_0h_1-k_1h_0>0$, $\nu=|h_1|$, and Lipschitz constant $K=|k_1|$. In view of Example~\ref{ex:CEV}, we must choose $\mean<\sqrt{2b(0)}/\nu=\sqrt{2(k_0h_1-k_1h_0)}/{|h_1|}$, and then take ${n_0}>K$ large enough for~\eqref{eq:min b0 CEV} to be satisfied, for instance ${n_0}\geq \max(2K, 8|k_1|\nu^2\mu^2/b(0))$.

\subsection{Two-factor CIR model}

In a two-factor interest rate model, the interest rate process is defined as $r(t)=\delta_0+\delta_1 X_1(t)+\delta_2 X_2(t)$, where $\delta_0\geq0$ and $\delta_1,\,\delta_2>0$ (see \cite{shreve2004b}). In the two-factor canonical CIR interest rate model, the two-dimensional process $(X(t))_{t\geq0}$ evolves according to the general dynamics given in equation~\eqref{E:mainsde}, where the coefficients are
\begin{equation*}
b(x)=\left(\begin{array}{c} \beta_1-\lambdauu x_1+\lambdaud x_2 \\
\beta_2+\lambdadu x_1-\lambdadd x_2\end{array}\right)\quad
\text{and}\quad
\sigma(x)=\left(\begin{array}{cc} \sqrt{x_1}& 0 \\
\sqrt{x_2}\rho & \sqrt{x_2}\sqrt{1-\rho^2}\end{array}\right),
\end{equation*}
for $x=(x_1,x_2)\in\Em=\mathbb{R}_+^2$, and with initial condition $X(0)$ in $\mathbb{R}_+^2$. The parameters $\beta_1,\,\beta_2,\,\lambdauu,\,\lambdadd$ are positive constants, $\lambdaud,\,\lambdadu$ are nonnegative, and the instantaneous correlation~$\rho$ satisfies $-1<\rho<1$. Note that we are in the case $d=m=2$. We define, for $x\in\Em$,
\begin{equation}\label{eq:volatilite cas CIR 2D}
\tilde{\sigma}(x)=\left(\begin{array}{cc} \sqrt{x_1}& 0 \\
0 & \sqrt{x_2}\end{array}\right)\quad\text{and}\quad
\Sigma=\left(\begin{array}{cc} 1& \rho \\
\rho & 1\end{array}\right)\,.
\end{equation}
Very similarly as in Example~\ref{ex:exemple du CIR une dim}, we may conclude that Condition~\ref{A:cond2} is satisfied as soon as
\begin{equation}\label{eq:resultat intermediaire CIR 2D}
{n_0}>\max(\lambdauu,\lambdadd)\quad\text{and}\quad0<\mu_i<2\sqrt{\beta_i\left(1-\frac{\lambda_{ii}}{{n_0}}\right)},\quad i=1,2,
\end{equation}
and the $2\times1$ random vector $\varepsilon$ has nonnegative components. Note that if $\varepsilon$ is a nonnegative random vector with  $\tilde{\esp}[\varepsilon]=\mean$, and covariance matrix~$\Sigma$ specified in~\eqref{eq:volatilite cas CIR 2D}, then it must be true that
\begin{equation}\label{eq:condition sur rho 2D}
-\rho\leq\mu_1\mu_2.
\end{equation}
Conversely, under condition~\eqref{eq:condition sur rho 2D}, one may construct a random vector $\varepsilon$ with nonnegative components, $\tilde{\esp}[\varepsilon]=\mean$, and covariance matrix~$\Sigma$ (specified in~\eqref{eq:volatilite cas CIR 2D}). As a result, it is possible to choose ${n_0}$ and $\mu_1\,,\mu_2$  such that~\eqref{eq:resultat intermediaire CIR 2D} and Condition~\ref{A:cond2} are satisfied as long as $-\rho<4\sqrt{\beta_1\beta_2}$.

\subsection{
Stochastic volatility models}\label{ex:GARCH11}
Let $\alpha,\lambda,\nu$ be positive, $\beta\in\mathbb{R}$, and $\rho\in(-1,1)$. Consider the GARCH(1,1) stochastic volatility model
\begin{align}
d V(t)\;&=\;(\alpha-\lambda V(t))dt +\nu V(t) d\tilde{W}_1(t)\,;\label{eq:Garch volatility process}\\
d S(t)\;&=\;S(t)\,(\beta dt+\sqrt{V(t)}d\tilde{W}_2(t))\,,\notag
\end{align}
with $V(0)=v_0>0$ and $S(0)=s_0>0$, and where $\tilde{W}_1$ and $\tilde{W}_2$ are two standard Brownian motions with instantaneous correlation~$\rho$. Even though both processes are nonnegative, we ease the situation by working (as in~\cite{lordetal2010}) with $\log S$ instead of~$S$, which does not need to remain positive. More precisely, we consider equation~\eqref{E:mainsde} with $d=2$ but $m=1$, and
\begin{equation*}
b(x)=\left(
\begin{array}{c}
\alpha - \lambda x_1 \\ \beta-\frac{x_1}{2}
\end{array}
\right)\,,\quad
\sigma(x)=\left(\begin{array}{cc}
\nu x_1 & 0 \\ \rho \sqrt{x_1} & \sqrt{1-\rho^2}\sqrt{x_1}
\end{array}\right),
\end{equation*}
for $x=(x_1,x_2)^\top$, and with $X(0)=(v_0,\log(s_0))^\top$. Define
\begin{equation}\label{eq:sigma tilde et Sigma pour le GARCH}
\tilde{\sigma}(x)=\left(\begin{array}{cc}
\nu x_1 & 0 \\ 0 & \sqrt{x_1}
\end{array}\right)\quad\text{and}\quad \Sigma=\left(\begin{array}{cc} 1& \rho \\
\rho & 1\end{array}\right)\,.
\end{equation}
Note that~\eqref{eq:nonnegativite des epsilon} in Condition~\ref{A:cond2} holds as soon as $\varepsilon_1$ is a nonnegative random variable ($\varepsilon_2$ may take both positive and negative values). For Condition~\ref{A:cond2} to be satisfied, it remains to make sure that
$$
\inf_{x_1\geq0}\left( x_1+\frac{\alpha-\lambda x_1}{n}-\frac{\nu x_1}{\sn}\mu_1\right)\;\geq0\,,\quad\text{for all }n\geq {n_0}\,.
$$
This will be the case if we choose ${n_0}>\lambda$ and $0<\mu_1<\frac{\sqrt{{n_0}}}{\nu}\left(1-\frac{\lambda}{{n_0}}\right)$\,. Note that, for $n\geq {n_0}$, the approximating scheme for $(X_1(t))_{t\geq0}=(V(t))_{t\geq0}$ given by~\eqref{E:algo} may be written as
$$
 Y_{1,n}(k+1)\;=\;\frac{\alpha}{n}+Y_{1,n}(k)\left(1-\frac{\lambda}{n}-\frac{\nu}{\sn}\mu_1\right)+\frac{\nu}{\sn}
 Y_{1,n}(k)\varepsilon_{1,k}^{(n)}\,,
$$
(where $Y_n=(Y_{1,n},Y_{2,n})^\top$ and $\varepsilon^{(n)}_k=(\varepsilon^{(n)}_{1,k},\varepsilon^{(n)}_{2,k})^\top$) and $Y_{1,n}(0)=v_0$. This is commonly refered to as a GARCH(1,1) discrete process, with the particularity that the family of random variables $\{\varepsilon_{1,k}^{(n)}\}_{k\geq0}$ are i.i.d. and follow any nonnegative distribution with mean $\mu_1$ and unit variance. Hence, convergence of the GARCH(1,1) discrete process to its continuous counterpart is a special case of Theorem~\ref{T:main}.

Note that if the process in~\eqref{eq:Garch volatility process} were replaced with a CIR process, a very similar argument could be applied, which would lead to a scheme for Heston's stochastic volatility model.

\section{Numerical results}\label{sec:numerical}

In this section, we present the results of two numerical experiments putting the proposed discretization scheme to the test. Results pertaining to other discretization schemes are also provided in order to facilitate comparisons. In the first experiment, we shall be interested in the price of a bond in the CIR interest rate model. This is a good test case since there is a well-known closed form formula for the bond price allowing for bias computations, and the payoff is path-dependent allowing to test the ability of the method for pricing path-dependent derivatives. In the second experiment, we consider the price of a (plain vanilla) European call option in Heston's stochastic volatility model. This is another interesting test case as it involves the joint action of two correlated processes, the CIR process being one of them, while the pricing problem is still mathematically tractable enough to obtain theoretical prices (\cite{heston1993} determines the characteristic function of the stock log-price at maturity; Fourier inversion can thereafter be used to get the price of a European call option).

\subsection{Bond pricing in the CIR model}\label{sec:bond pricing CIR}

For the first experiment, we consider the CIR process of Example~\ref{ex:exemple du CIR une dim}, that is $(X(t))_{t\geq0}$ is a solution of~\eqref{E:mainsde} when $b(t,x)=\kappa(\beta-x)$ and $\sigma(t,x)=\nu\sqrt{x}$, where $\kappa$, $\beta$ and $\nu$ are positive constants. Assuming that $(X(t))_{t\geq0}$ models the interest rate, the value of a bond that will pay its face value of 1000 dollars in $T$ years is worth $1000\times\esp\left[\exp\left( -\int_0^T X(s)\,ds\right)\right]$ today. Our goal is to get an approximation of this price using Monte Carlo simulations together with various discretization schemes and compare them to the true value of the bond; see, e.g., \cite[Proposition~6.2.5]{lambertonlapeyre1991} for the closed-form formula.

It is well known that the CIR process will never hit the origin, that is $\p(X(t)>0\text{ for all }t\geq0)=1$ if $x_0>0$, when the model parameters satisfy $\sqrt{2\kappa\beta}\geq\nu$, while it eventually hits the origin with probability one when $\sqrt{2\kappa\beta}<\nu$; see, e.g., \cite[Proposition~6.2.4]{lambertonlapeyre1991}. Since we want to assess how the different discretization schemes manage the zero boundary condition, it is natural to test the methods for parameters for which the boundary problem will theoretically arise. In fact, we put $\kappa=0.5$ and $\beta=0.04$, so that $\sqrt{2\kappa\beta}=0.2$, and consider two values exceeding $0.2$ for the volatility coefficient: $\nu=0.3$ and $\nu=1$. In the latter case, the zero boundary problem is magnified by the larger volatility and so it should be intuitively more difficult to get accurate prices. Tables~\ref{tab:CIR low volatility} and~\ref{tab:CIR high volatility} show, for different values of~$n$ and Euler-type discretization schemes, the bias and the margin of error at the 95\% confidence level for the Monte Carlo prices, based on one million trajectories of the CIR process (i.e., $N=1\,000\,000$ in the notation of Remark~\ref{rem:simulations Monte Carlo}). For the reader's convenience, those biases that are not significantly different from zero at the 95\% level appear in bold fonts. In our discretization scheme (the one defined in~\eqref{eq:schema une dim}), we have set for $\varepsilon$ the two-valued distribution, one of whose values is zero, with average $\mu=0.8$ in the low volatility case and $\mu=0.28$ in the high volatility case (recall~\eqref{eq:bernoulli pour le MBG}, where this Bernoulli-type distribution is explicitly given). Note that with these values of $\mu$ the condition $0<\mean\leq\frac2\nu \left( \kappa\beta(1-(\kappa/n_0)) \right)^{1/2}$ from Example~\ref{ex:exemple du CIR une dim} is satisfied with respectively ${n_0}=4$ and ${n_0}=50$ for the two sets of parameters. This choice of distribution is motivated by the fact that the Bernoulli distribution is the simplest there is, which results in reduced simulation time. Note that the approximation of the payoff term $\int_0^T X(s)\,ds$ requires an approximation of the path based on the discretization scheme for the entire interval $[0,T]$. For the scheme (b4), the trajectories have been interpolated in between the time-steps using $X_n(s)=|Y_n(\lfloor ns \rfloor)|$, $s\in[0,T]$, since \citep[(4), (26), Theorem~4.1]{highammao2005} prove that the bond approximate prices converge to the right price as $n$ goes to infinity under this setup. For all other schemes, the linear interpolation defined in~\eqref{eq:linear interpolation} is used. Since the payoff function $g(x)=1000\times\exp(-\int_0^T x(s)\,ds)$, for $x\in C_{\reels_+}(\reels_+)$, is continuous and bounded, then Proposition~\ref{prop:conditions generales pour la convergence des esperances} implies that Monte Carlo bond prices from our scheme converge to the right price as $n$ goes to infinity.

\begin{table}[h]
\begin{center}
\caption{Bond pricing in the low volatility case}
\begin{tabular}{|c|c|c|c|c|c|}
  \hline
  $n=$ steps$/$year & Bernoulli & (b1) & (b2) & (b3) & (b4) \\\hline\hline
  4 & $\begin{array}{c} {0.204} \\ (0.123) \end{array}$ &    
      $\begin{array}{c} {0.258} \\ (0.126) \end{array}$ &    
      $\begin{array}{c} {1.367} \\ (0.129) \end{array}$ &        
      $\begin{array}{c} {-12.825} \\ (0.113) \end{array}$ &        
      $\begin{array}{c} {-9.174} \\ (0.107) \end{array}$ \\\hline 
  6 & $\begin{array}{c} {\bf 0.080} \\ (0.122) \end{array}$ &        
      $\begin{array}{c} {0.168 } \\ (0.124) \end{array}$ &        
      $\begin{array}{c} {0.842 } \\ (0.125) \end{array}$ &        
      $\begin{array}{c} {-9.305 } \\ (0.114) \end{array}$ &        
      $\begin{array}{c} {-6.124 } \\ (0.110) \end{array}$ \\\hline 
  8 & $\begin{array}{c} {\bf 0.040} \\ (0.121) \end{array}$ &        
      $\begin{array}{c} {0.188 } \\ (0.122) \end{array}$ &        
      $\begin{array}{c} {0.505} \\ (0.123) \end{array}$ &        
      $\begin{array}{c} {-7.314 } \\ (0.115) \end{array}$ &        
      $\begin{array}{c} {-4.696} \\ (0.112) \end{array}$ \\\hline 
  10 & $\begin{array}{c} {0.154 } \\ (0.120) \end{array}$ &        
      $\begin{array}{c} {\bf 0.109} \\ (0.121) \end{array}$ &        
      $\begin{array}{c} {0.337 } \\ (0.122) \end{array}$ &        
      $\begin{array}{c} {-6.070 } \\ (0.115) \end{array}$ &        
      $\begin{array}{c} {-3.744} \\ (0.114) \end{array}$ \\\hline 
  20 & $\begin{array}{c} {\bf 0.060} \\ (0.119) \end{array}$ &        
      $\begin{array}{c} {\bf 0.057} \\ (0.120) \end{array}$ &        
      $\begin{array}{c} {0.166 } \\ (0.120) \end{array}$ &        
      $\begin{array}{c} {-3.545 } \\ (0.116) \end{array}$ &        
      $\begin{array}{c} {-1.860 } \\ (0.116) \end{array}$ \\\hline 
  40 & $\begin{array}{c} {\bf -0.062}^\ast \\ (0.119) \end{array}$ &        
      $\begin{array}{c} {\bf 0.057}^\ast \\ (0.119) \end{array}$ &        
      $\begin{array}{c} {\bf -0.051} \\ (0.120) \end{array}$ &        
      $\begin{array}{c} {-2.257 } \\ (0.117) \end{array}$ &        
      $\begin{array}{c} {-0.971 } \\ (0.117) \end{array}$ \\\hline 
  80 & $\begin{array}{c} {\bf 0.029}^\ast \\ (0.118) \end{array}$ &        
      $\begin{array}{c} {\bf -0.008}^\ast \\ (0.119) \end{array}$ &        
      $\begin{array}{c} {\bf -0.082}^\ast \\ (0.119) \end{array}$ &        
      $\begin{array}{c} {-1.558 } \\ (0.117) \end{array}$ &        
      $\begin{array}{c} {-0.373 } \\ (0.118) \end{array}$ \\\hline 
  160 & $\begin{array}{c} {\bf 0.008}^\ast \\ (0.119) \end{array}$ &        
      $\begin{array}{c} {\bf -0.000}^\ast \\ (0.119) \end{array}$ &        
      $\begin{array}{c} {\bf 0.005}^\ast \\ (0.119) \end{array}$ &        
      $\begin{array}{c} {-0.991 } \\ (0.118) \end{array}$ &        
      $\begin{array}{c} {-0.261 } \\ (0.118) \end{array}$ \\\hline\hline 
rate of conv. & $0.54$ & $0.94$ & $1.41$ & $1.99$ & $2.85$ \\\hline
\end{tabular}
\label{tab:CIR low volatility}\\
\begin{minipage}{12cm}
Bias and margin of error at the 95\% confidence level (in parentheses).\\
CIR process parameters: $\kappa=0.5$, $\beta=x_0=0.04$, $\nu=0.3$.\\
Bond parameters: $T=2$ years and face value of 1000.\\
True bond value: 925.258.\\
Mean for $\varepsilon$: $\mu=0.8$.
\end{minipage}
\end{center}
\end{table}

\begin{figure}
\caption{Bias comparison in the low volatility case}\label{fig:CIR low volatility}
\includegraphics[scale=0.5]{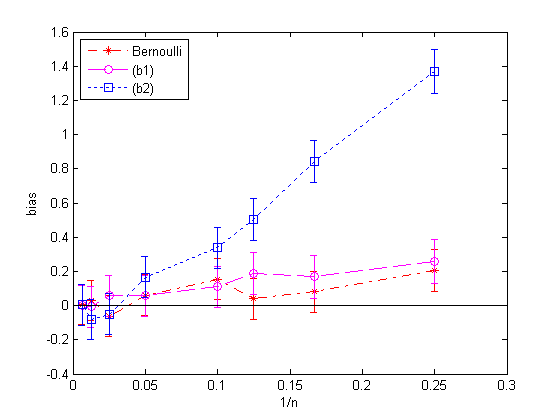}
\end{figure}

\begin{table}[h]
\begin{center}
\caption{Bond pricing in the high volatility case}
\begin{tabular}{|c|c|c|c|c|c|}
  \hline
  $n=$ steps$/$year & Bernoulli & (b1) & (b2) & (b3) & (b4) \\\hline\hline
  50 & $\begin{array}{c} {-0.678 } \\ (0.249) \end{array}$ &        
      $\begin{array}{c} {2.044 } \\ (0.270) \end{array}$ &        
      $\begin{array}{c} {4.720 } \\ (0.271) \end{array}$ &        
      $\begin{array}{c} {-117.019 } \\ (0.318) \end{array}$ &        
      $\begin{array}{c} {-108.046 } \\ (0.311) \end{array}$ \\\hline 
  100 & $\begin{array}{c} {-0.329 } \\ (0.250) \end{array}$ &        
      $\begin{array}{c} {0.798 } \\ (0.263) \end{array}$ &        
      $\begin{array}{c} {2.086 } \\ (0.263) \end{array}$ &        
      $\begin{array}{c} {-98.960} \\ (0.313) \end{array}$ &        
      $\begin{array}{c} {-90.511} \\ (0.307) \end{array}$ \\\hline 
  200 & $\begin{array}{c} {-0.435 } \\ (0.252) \end{array}$ &        
      $\begin{array}{c} {0.442 } \\ (0.259) \end{array}$ &        
      $\begin{array}{c} {1.263} \\ (0.257) \end{array}$ &        
      $\begin{array}{c} {-84.907} \\ (0.308) \end{array}$ &        
      $\begin{array}{c} {-76.459} \\ (0.302) \end{array}$ \\\hline 
  400 & $\begin{array}{c} {-0.368} \\ (0.253) \end{array}$ &        
      $\begin{array}{c} {0.278} \\ (0.257) \end{array}$ &        
      $\begin{array}{c} {0.453} \\ (0.257) \end{array}$ &        
      $\begin{array}{c} {-74.073} \\ (0.304) \end{array}$ &        
      $\begin{array}{c} {-66.008} \\ (0.298) \end{array}$ \\\hline 
  800 & $\begin{array}{c} {\bf -0.207} \\ (0.253) \end{array}$ &        
      $\begin{array}{c} {\bf 0.007} \\ (0.256) \end{array}$ &        
      $\begin{array}{c} {\bf 0.144} \\ (0.256) \end{array}$ &        
      $\begin{array}{c} {-65.310 } \\ (0.300) \end{array}$ &        
      $\begin{array}{c} {-57.324 } \\ (0.294) \end{array}$ \\\hline 
  1600 & $\begin{array}{c} {\bf -0.064}^\ast \\ (0.253) \end{array}$ &        
      $\begin{array}{c} {\bf -0.067}^\ast \\ (0.256) \end{array}$ &        
      $\begin{array}{c} {0.402} \\ (0.255) \end{array}$ &        
      $\begin{array}{c} {-58.399} \\ (0.297) \end{array}$ &        
      $\begin{array}{c} {-50.397} \\ (0.290) \end{array}$ \\\hline\hline 
rate of conv. & $0.33$ & $1.77$ & $0.88$ & $0.20$ & $0.22$ \\\hline
\end{tabular}
\label{tab:CIR high volatility}\\
\begin{minipage}{12cm}
Bias and margin of error at the 95\% confidence level (in parentheses).\\
CIR process parameters: $\kappa=0.5$, $\beta=x_0=0.04$, $\nu=1$.\\
Bond parameters: $T=2$ years and face value of 1000.\\
True bond value: 940.024.\\
Mean for $\varepsilon$: $\mu=0.28$.
\end{minipage}
\end{center}
\end{table}

\begin{figure}
\caption{Bias comparison in the high volatility case}\label{fig:CIR high volatility}
\includegraphics[scale=0.5]{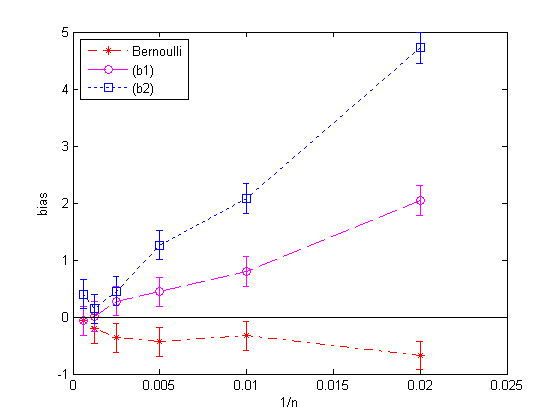}
\end{figure}

In Table~\ref{tab:CIR low volatility} and Table~\ref{tab:CIR high volatility} we see that biases are consistently less than a dollar for our scheme (in the Bernoulli column) and schemes (b1) and (b2), and for the values of~$n$ shown; this is quite small in comparison to the true bond prices which are respectively 925 and 940 dollars. The biases and intervals at the 95\% confidence level of those three schemes are then compared graphically in Figure~\ref{fig:CIR low volatility} and Figure~\ref{fig:CIR high volatility}. As anticipated, for all methods it is more difficult to evaluate the bond price in the high volatility case. Indeed, in the latter case much more time steps per year are required to get an approximate price which is not significantly different from the true price (at the 95\% confidence level, based on one million trajectories).

Table~\ref{tab:CIR low volatility} and Table~\ref{tab:CIR high volatility} were designed primarily to compare biases, some of which are quite different in sizes for the same~$n$. As a complement of information, their last rows provide a rough estimate of the weak order of convergence. Recall that these tables are based on one million trajectories, which is actually not sufficient to get precise numerical orders of convergence. 
Here, the rate of convergence is estimated to be the slope (in absolute value), when linearly regressing $\log(|\text{bias}|)$ on $\log(n)$ (and a constant term). For methods that do very well, note that several values of the bias shown in the tables are not significantly different from zero. Hence, these values carry very little information on the exact size of the deviation between $\tilde{\esp}[g(X_n)]$ and the real price $\esp[g(X)]$; they are essentially noise. Hence, some of these values (marked by an asterisk in the tables) have not been taken into consideration in the regressions. By using much larger numbers of trajectories, the actual deviation between $\tilde{\esp}[g(X_n)]$ and the real price has been estimated with much more accuracy in the case of the proposed scheme. The results are reported in Table~\ref{tab:CIR Bernoulli order of convergence}. In the low volatility case, the estimated order of convergence of 1.2 lies between the rough estimates obtained for methods (b1) and (b2), namely 0.94 and 1.41. However, the estimated order of convergence of 0.36 is smaller than both the rough estimates for methods (b1) and (b2) in the high volatility case, so the actual order of convergence might be slower for the proposed method. Calculation of the real order of weak convergence could be the subject of future work. Note that regression of the root mean square error (RMSE) instead of the log-bias would have produced similar results. 

\begin{table}[h]
\begin{center}
\caption{Bond pricing using the Bernoulli distribution}
\begin{tabular}{|c|r|c|c|c|c|r|c|c|c|}
  \hline
  \multicolumn{5}{|c|}{low volatility case} & \multicolumn{5}{|c|}{high volatility case}\\\hline
  $n$ & \multicolumn{1}{|c|}{$N$} & bias & margin & RMSE & $n$ & \multicolumn{1}{|c|}{$N$} & bias & margin & RMSE\\\hline\hline
  4  & $4\times 10^6$  & $0.1951$ & $0.0616$ & $0.1977$ & 50  & $4\times 10^6$  & $-0.4800$ & $0.1243$ & $0.4842$ \\
  8  & $16\times 10^6$ & $0.0675$ & $0.0302$ & $0.0692$ & 100 & $8\times 10^6$  & $-0.3940$ & $0.0884$ & $0.3965$ \\
  10 & $25\times 10^6$ & $0.0801$ & $0.0240$ & $0.0810$ & 200 & $16\times 10^6$ & $-0.2883$ & $0.0628$ & $0.2901$ \\
  20 & $100\times 10^6$& $0.0197$ & $0.0119$ & $0.0206$ & 400 & $32\times 10^6$ & $-0.2583$ & $0.0447$ & $0.2593$ \\
  40 & $400\times 10^6$& $0.0134$ & $0.0059$ & $0.0137$ & 800 & $64\times 10^6$ & $-0.1694$ & $0.0317$ & $0.1701$ \\\hline\hline
  \multicolumn{2}{|r|}{rate of conv.}  & $1.209$  &  & $1.201$ & \multicolumn{2}{|r|}{} & $0.362$ & & $0.363$ \\\hline
\end{tabular}
\label{tab:CIR Bernoulli order of convergence}\\
\begin{minipage}{15cm}
Bias, margin of error at the 95\% confidence level, and RMSE for different combinations of~$n$ (the number of steps per year) and~$N$ (the number of trajectories).\\
The parameters are the same as in Table~\ref{tab:CIR low volatility} (resp. Table~\ref{tab:CIR high volatility}) in the low (resp. high) volatility case.
\end{minipage}
\end{center}
\end{table}

From this first numerical experiment, one can safely conclude that our scheme is very competitive when it comes to discretizing the CIR dynamic and pricing a path-dependent derivative, in both low and high volatility environments.

\subsection{Pricing of a European call in Heston's model} For the second experiment, we consider Heston's stochastic volatility model (in the risk-neutral world), in which the squared volatility~$V$ and stock price~$S$ evolve according to
\begin{align*}
d V(t)\;&=\;\kappa(\beta-V(t))dt +\nu \sqrt{V(t)} d\tilde{W}_1(t)\,,\\
d S(t)\;&=\;S(t)\,(r dt+\sqrt{V(t)}d\tilde{W}_2(t))\,,
\end{align*}
with $V(0)=v_0$ and $S(0)=s_0$. Here, $\kappa$, $\beta$, $\nu$, $r$, $v_0$ and $s_0$ are positive constants, and $\tilde{W}_1$ and $\tilde{W}_2$ are two standard Brownian motions with instantaneous correlation~$\rho$. Our goal is to approximate the price of a European plain vanilla option whose value is $\esp[e^{-r T}(S(T)-K)^+]$. For ease of comparison, we use the same set of model parameters as in experiment SV-I in \cite{lordetal2010} and experiment two in \cite{broadiekaya2006}. These parameters are shown below Table~\ref{tab:Heston SV-I}. Note that $\sqrt{2\kappa\beta}=0.6<\nu=1$, hence the volatility process eventually hits zero with probability one, and the way the discretization schemes handle the zero boundary is strongly put to the test.

As in \cite{lordetal2010}, we work with $\log S$ instead of~$S$ to simplify matters. In other words, we consider equation~\eqref{E:mainsde} with $d=2$, $m=1$, and
\begin{equation*}
b(x)=\left(
\begin{array}{c}
\kappa(\beta - x_1) \\ r-\frac{x_1}{2}
\end{array}
\right)\,,\quad
\sigma(x)=\left(\begin{array}{cc}
\nu \sqrt{x_1} & 0 \\ \rho \sqrt{x_1} & \sqrt{1-\rho^2}\sqrt{x_1}
\end{array}\right),
\end{equation*}
for $x=(x_1,x_2)^\top\in E$, and with $X(0)=(v_0, \log s_0)^\top$.
Define $\tilde{\sigma}(x)$ as the $2\times2$ diagonal matrix $\tilde{\sigma}(x)=\text{diag}[\nu\sqrt{x_1},\sqrt{x_1}]$, and $\Sigma$ exactly as in~\eqref{eq:sigma tilde et Sigma pour le GARCH}. Then Condition~\ref{A:cond2} is satisfied if ${n_0}>\kappa$ and $\mu_1$ satisfies $0<\mu_1\leq\frac2\nu \left( \kappa\beta(1-(\kappa/n_0)) \right)^{1/2}$ (compare with Example~\ref{ex:exemple du CIR une dim} dealing with the CIR in isolation), and if the random vector $\varepsilon=(\varepsilon_1,\varepsilon_2)^\top$ is such that $\varepsilon_1$ is a nonnegative variable with mean $\mu_1$ and variance~1, $\varepsilon_2$ is any variable with variance~1, and $\varepsilon_1$ and $\varepsilon_2$ have correlation~$\rho$. In order to produce such a random vector for the simulations, we independently generate $\varepsilon_1$ and $\varepsilon_3$ as the Bernoulli-type variables with respective means $\mu_1=0.657$ and $\mu_3=1$ and variance~1 (as in~\eqref{eq:bernoulli pour le MBG}), then we set $\varepsilon_2=\rho\;\varepsilon_1+\sqrt{1-\rho^2}\;\varepsilon_3$. The payoff functional $g(x)=e^{-rT}(e^{x_2(T)}-K)^+$, for $x=(x_1,x_2)^\top\in C_E(\reels_+)$, is continuous, as the payoff function of the corresponding put option would be. Convergence of the Monte Carlo prices to the right price when using our scheme (in~\eqref{E:algo}) hence follows from Proposition~\ref{prop:conditions generales pour la convergence des esperances} and the put-call parity.

All results in Table~\ref{tab:Heston SV-I} are based on one million joint trajectories of the squared-volatility (following the CIR dynamics) and the log-price (i.e., $N=1\,000\,000$). No matter the scheme, the discretization of the log-price involes the square-root of the discretized CIR process, which might not be well defined if the discretized CIR becomes negative. The latter might happen if we use (b1), (b2) or (b4) on the CIR process. Following \cite{lordetal2010}, we fix this problem by using, for the discretization of the log-price, the positive part of the discretized CIR for schemes (b1) and (b2), and the absolute value of the discretized CIR for scheme (b4). Also, note that \cite{kahletal2008} establish that the implicit Milstein method produces a well-defined and nonnegative approximate path for the CIR process only when the parameters satisfy $\sqrt{2\kappa\beta}\geq\nu$. Our parameter values do not satisfy this condition, which is our main reason for not including such implicit schemes in this section. Indeed, with our parameters, the implicit scheme would have to be combined with another fix to work. However, the interested reader may consult Table~4 in \cite{lordetal2010}, which is similar to Table~\ref{tab:Heston SV-I} and includes results for the Milstein implicit scheme fixed to remain nonnegative.

\begin{table}[h]
\begin{center}
\caption{Pricing of a European call option}
\begin{tabular}{|c|c|c|c|c|c|}
  \hline
  $n=$ steps$/$year & Bernoulli & (b1) & (b2) & (b3) & (b4) \\\hline\hline
    5 & $\begin{array}{c} {-0.121 } \\ (0.108) \end{array}$ &        
      $\begin{array}{c} {1.868 } \\ (0.128) \end{array}$ &        
      $\begin{array}{c} {0.359 } \\ (0.117) \end{array}$ &        
      $\begin{array}{c} {8.318 } \\ (0.194) \end{array}$ &        
      $\begin{array}{c} {6.995 } \\ (0.188) \end{array}$ \\\hline 
10 & $\begin{array}{c} {\bf -0.087} \\ (0.111) \end{array}$ &        
      $\begin{array}{c} {0.948 } \\ (0.120) \end{array}$ &        
      $\begin{array}{c} {0.185 } \\ (0.115) \end{array}$ &        
      $\begin{array}{c} {6.055 } \\ (0.165) \end{array}$ &        
      $\begin{array}{c} {4.453 } \\ (0.158) \end{array}$ \\\hline 
20 & $\begin{array}{c} {\bf -0.061} \\ (0.0.112) \end{array}$ &        
      $\begin{array}{c} {0.500 } \\ (0.116) \end{array}$ &        
      $\begin{array}{c} {0.137 } \\ (0.113) \end{array}$ &        
      $\begin{array}{c} {4.419 } \\ (0.148) \end{array}$ &        
      $\begin{array}{c} {2.733 } \\ (0.140) \end{array}$ \\\hline 
40 & $\begin{array}{c} {\bf -0.070}^\ast \\ (0.112) \end{array}$ &        
      $\begin{array}{c} {0.147 } \\ (0.115) \end{array}$ &        
      $\begin{array}{c} {\bf -0.013} \\ (0.114) \end{array}$ &        
      $\begin{array}{c} {3.181 } \\ (0.139) \end{array}$ &        
      $\begin{array}{c} {1.649 } \\ (0.129) \end{array}$ \\\hline 
80 & $\begin{array}{c} {\bf -0.013}^\ast \\ (0.113) \end{array}$ &        
      $\begin{array}{c} {0.126 } \\ (0.114) \end{array}$ &        
      $\begin{array}{c} {\bf 0.057} \\ (0.113) \end{array}$ &        
      $\begin{array}{c} {2.377 } \\ (0.131) \end{array}$ &        
      $\begin{array}{c} {1.074 } \\ (0.123) \end{array}$ \\\hline 
160 & $\begin{array}{c} {\bf 0.093}^\ast \\ (0.114) \end{array}$ &        
      $\begin{array}{c} {\bf 0.059} \\ (0.114) \end{array}$ &        
      $\begin{array}{c} {\bf 0.030}^\ast \\ (0.113) \end{array}$ &        
      $\begin{array}{c} {1.795 } \\ (0.127) \end{array}$ &        
      $\begin{array}{c} {0.651 } \\ (0.119) \end{array}$ \\\hline\hline 
rate of conv. & $0.50$ & $1.01$ & $0.92$ & $0.45$ & $0.69$ \\\hline
\end{tabular}
\label{tab:Heston SV-I}\\
\begin{minipage}{14cm}
Bias and margin of error at the 95\% confidence level (in parentheses).\\
Volatility parameters: $\kappa=2$, $\beta=v_0=0.09$, $\nu=1$.\\
Option and stock price parameters: $T=5$ years, $s_0=100$, $K=100$, $r=0.05$.\\
Correlation parameter: $\rho=-0.3$.\\
True option price: 34.9998.\\
Means for $\varepsilon_1$ and $\varepsilon_3$: $\mu_1=0.657$ and $\mu_3=1$.
\end{minipage}
\end{center}
\end{table}

\begin{figure}
\caption{Bias comparison for the option price}\label{fig:Heston SV-I}
\includegraphics[scale=0.5]{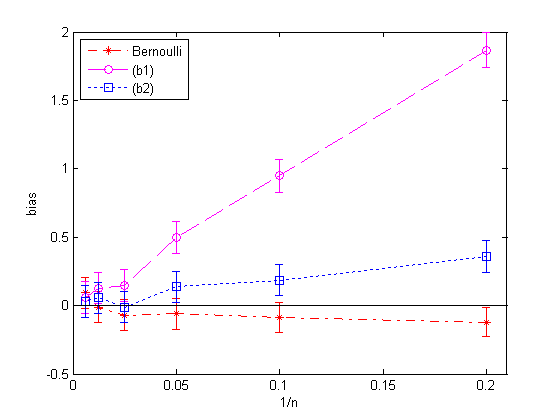}
\end{figure}

Table~\ref{tab:Heston SV-I} confirms what has been observed in the bond pricing experiment: the first three schemes outperform schemes (b3) and (b4) in terms of biases. The first three schemes are also compared graphically in Figure~\ref{fig:Heston SV-I}. \cite{lordetal2010} numerically compare several discretization schemes for option pricing in Heston's model and conclude that scheme (b2) is very efficient. From Table~\ref{tab:Heston SV-I} and Figure~\ref{fig:Heston SV-I}, one can clearly conclude that our scheme is competitive and can be compared favourably with scheme (b2).

As in Section~\ref{sec:bond pricing CIR}, orders of convergence are estimated. Rough estimates in Table~\ref{tab:Heston SV-I} may be compared with those in Table~4 in \cite{lordetal2010}, who claim that it is quite hard in this case to properly estimate the order of convergence even when using 10 millions trajectories. However, care is advisable when interpreting the estimates of the order of convergence in Table~\ref{tab:Heston SV-I}, especially for our method. Indeed, the biases presented in Table~\ref{tab:Heston SV-I} are not significantly different from zero (for our scheme only) for values of $n$ as small as 10 and 20; moreover, the bias for $n=5$ is the lowest. A more accurate estimate is therefore provided in Table~\ref{tab:Heston Bernoulli order of convergence}. Although the estimated rate of convergence of 0.65 is smaller than those estimated for methods (b1) and (b2), it has to be kept in mind that (as is made clear in Figure~\ref{fig:Heston SV-I}) the biases for the proposed method are small overall in comparison with other methods. Hence, one may certainly claim that the proposed method is a competitive alternative to other methods, even more since Figure~\ref{fig:CIR low volatility} and Figure~\ref{fig:CIR high volatility} led to a similar conclusion in another meaningful experiment. 

\begin{table}[h]
\begin{center}
\caption{Option pricing using the Bernoulli distribution}
\begin{tabular}{|c|r|c|c|c|}
  \hline
  $n$ & \multicolumn{1}{|c|}{$N$} & bias & margin & RMSE \\\hline\hline
  5  & $5\times10^6$  & $-0.1144$ & $0.0480$ & $0.1169$ \\
  10 & $10\times10^6$ & $-0.0911$ &	$0.0350$ & $0.0929$ \\
  20 & $20\times10^6$ & $-0.0435$ &	$0.0251$ & $0.0453$ \\
  40 & $40\times10^6$ & $-0.0452$ &	$0.0178$ & $0.0461$ \\
  80 & $80\times10^6$ & $-0.0227$ & $0.0126$ & $0.0236$ \\
  160& $160\times10^6$& $-0.0109$ & $0.0090$ & $0.0118$ \\\hline\hline
  \multicolumn{2}{|r|}{rate of conv.}  & $0.654$ & & $0.641$ \\\hline
\end{tabular}
\label{tab:Heston Bernoulli order of convergence}\\
\begin{minipage}{10cm}
Bias, margin of error at the 95\% confidence level, and RMSE for different combinations of~$n$ (the number of steps per year) and~$N$ (the number of trajectories).\\
The parameters are the same as in Table~\ref{tab:Heston SV-I}.
\end{minipage}
\end{center}
\end{table}

Finally, Tables~\ref{tab:CIR low volatility}, \ref{tab:CIR high volatility} and \ref{tab:Heston SV-I} clearly illustrate that the rate of convergence, as measured here, is not as good a measure of precision as it appears to be. In these tables, the error of the biases are comparable but our proposed method is much more precise in the sense that the bias is not significantly different from zero, even for small number of steps $n$.

\subsection{Conclusion} In summary, in addition to the simplicity of our scheme, its great flexibility and the fact that approximate prices are theoretically known to converge, it is observed to be numerically competitive with other existing schemes in two representative experiments. We recall that our scheme and our convergence results hold in more general diffusion models and for a wide variety of payoffs.

\begin{center}{\sc Acknowledgements}\end{center}

We thank anonymous referees and the Associate Editor for a careful reading and useful comments on an earlier version of this paper.

Funding in partial support of this work was provided by the Natural Sciences and Engineering Research Council of Canada, the Fonds qu\'eb\'ecois de la
recherche sur la nature et les technologies, and the Insti\-tut de finance math\'ematique de Montr\'eal.

\appendix

\section{Proof of Theorem~\ref{T:main}}

We now prove Theorem~\ref{T:main}. As discussed in Section~\ref{sec:convergence}, we will use the martingale problem formulation. We need the next result, which gives sufficient conditions for $A_nf$ defined in~\eqref{eq:generateur discret} to converges to $Af$, for all $f\in C_c^\infty(\reels^d)$, and then concludes that this is enough for the corresponding processes to converge as well. Throughout, $|\cdot|$ is the Euclidean norm.

\begin{prop}\label{prop:convergence d'une chaine de Markov vers une diffusion, resultat de Ethier et Kurtz}
Assume Conditions~\ref{A:cond0} and~\ref{A:cond1}. For each $n\geq {n_0}$, let $K_n(t,x,\Gamma)$ be a time-dependent transition function defined on $\tE\times\mathcal{B}(\reels^d)$ and such that $K_n(t,x,\Em)=1$ for all $(t,x)\in\tE$.
For $n\geq {n_0}$, set
\begin{equation}\label{eq:definition de bn}
b_n(t,x)\;:=\;n\int_{|y-x|\leq1} (y-x)K_n(t,x,dy)
\end{equation}
and
\begin{equation}\label{eq:definition de an}
a_n(t,x)\;:=\;n\int_{|y-x|\leq1} (y-x)(y-x)^\top K_n(t,x,dy)\,,
\end{equation}
for each $(t,x)\in\tE$.
Assume further that, for any $r>0$ and $\epsilon>0$, the following sequences tend to zero as $n$ goes to infinity:
\begin{equation}\label{eq:an et bn sont pres de a et b}
\sup_{\substack{(t,x)\in \tE\\|(t,x)|\leq r}}|b_n(t,x)-b(t,x)|\,,\quad\sup_{\substack{(t,x)\in \tE\\|(t,x)|\leq r}}|a_n(t,x)-a(t,x)|\,,
\end{equation}
and
\begin{equation}\label{eq:chaine de Markov ne saute pas beaucoup}
\sup_{\substack{(t,x)\in \tE\\|(t,x)|\leq r}} n\, K_n(t,x,\,\{y\,:\,|y-x|\geq\epsilon\})\,.
\end{equation}
Let $({Y}_n(k))_{k\geq0}$ be a Markov chain, with ${Y}_n(0)=x_0\in\Em$, and transitions governed by~$K_n$ through equation~\eqref{eq:transitions de Ynk}. Then
$$
 \lim_{n\rightarrow\infty}\sup_{\substack{(t,x)\in \tE\\|(t,x)|\leq r}} |A_nf(t,x)-Af(t,x)|=0,
$$
for all $r>0$ and $f\in C_c^\infty(\reels^d)$, where $A_n$ is defined in terms of~$K_n$ in~\eqref{eq:generateur discret} and $A$ is defined in~\eqref{E:gen}. Moreover, the sequence of processes $(X_n)_{n\geq {n_0}}$, defined from $(Y_n)_{n\geq {n_0}}$ through~\eqref{eq:linear interpolation}, converges in distribution to the solution of the martingale problem for $(A,x_0)$.
\end{prop}

That proposition corresponds to results in \citep[Section~11.2]{stroockvaradhan1979}, restricted to $E$. See also Corollary~7.4.2, in conjunction with Proposition~3.10.4 in~\cite{ethierkurtz1986}. Note that the functions~$a$ and~$b$ postulated in Section 11.2 in~\cite{stroockvaradhan1979} or Corollary~7.4.2 in~\cite{ethierkurtz1986} are time-homogeneous. To extend these results to the case where~$a$ and~$b$ are time-dependent, one only needs to append the time to the state variable. In other words, consider the $(d+1)$-dimensional state process $\check{X}(t):=(X(t)^\top, t)^\top$. Then $X$ is a solution to the SDE~\eqref{E:mainsde} if and only if $\check{X}$ is a solution to $d\check{X}_t=\check{b}(\check{X}_t)dt+\check{\sigma}(\check{X}_t)d\check{W}(t)$ with initial condition $\check{X}(0)=(x_0^\top,0)^\top$, upon defining
\begin{equation*}
\check{b}\left(\begin{array}{c}x\\t\end{array}\right):=\left(\begin{array}{c}
b(t,x)\\1
\end{array}\right)\quad\text{and}\quad
\check{\sigma}\left(\begin{array}{c}x\\t\end{array}\right):=\left(\begin{array}{cc}
\sigma(t,x) & 0 \\ 0 & 0
\end{array}\right)\,,
\end{equation*}
where $\check{W}$ stands for a $(d+1)$-dimensional standard Brownian motion.

Our goal is now to verify that Proposition~\ref{prop:convergence d'une chaine de Markov vers une diffusion, resultat de Ethier et Kurtz} applies when the transition function~$K_n$ is the one defined at~\eqref{eq:transition function}. From Condition~\ref{A:cond2}, it is clear that $K_n(t,x,\Em)=1$ for all $(t,x)\in\Em$. 
The two following lemmas establish that the three sequences given in~\eqref{eq:an et bn sont pres de a et b} and~\eqref{eq:chaine de Markov ne saute pas beaucoup} tend to zero as $n$ goes to infinity, which is all that is needed to get the desired convergence. To lighten the notation, we shall write $Z:=\varepsilon-\mean$.

\begin{lem}\label{lemme:resultats techniques sur An}
Assume Conditions~\ref{A:cond0}, \ref{A:cond1} and~\ref{A:cond2}. For $t\geq0$, $x\in\Em$ and $\epsilon>0$, define
\begin{equation*}
D_n(t,x,\epsilon)\;:=\;\left\{ \left| \frac1n b(t,x)+\frac{1}{\sn}\tilde{\sigma}(t,x)Z \right|\,>\,\epsilon\right\}\,.
\end{equation*}
Then, for any $r>0$ and $\epsilon>0$, we have, as $n$ goes to infinity, that
\begin{alignat}{3}
&\sup_{\substack{(t,x)\in \tE\\ |(t,x)|\leq r}} \tilde{\mathbb{E}}\left[|Z|^2\, \mathbb{I}_{D_n(t,x,\epsilon)}\right]&&\rightarrow0\,;\label{eq:Z carre un An va a zero}\\
&\sup_{\substack{(t,x)\in \tE\\ |(t,x)|\leq r}} n\,\tilde{\mathbb{P}}(D_n(t,x,\epsilon))&&\rightarrow0\,.\label{eq:n prob de An va a zero}
\end{alignat}
\end{lem}
\begin{proof}
For each $r>0$, define
\begin{equation*}
\Mb(r)\;:=\;\max_{\substack{(t,x)\in \tE\\ |(t,x)|\leq r}}|b(t,x)|\quad\text{and}\quad \Msigma(r)\;:=\;\max_{\substack{(t,x)\in \tE\\ |(t,x)|\leq r}}|\tilde{\sigma}(t,x)|\,.
\end{equation*}
Since $b$ and $\tilde{\sigma}$ are assumed to be continuous, we have $0\leq\Mb(r),\Msigma(r)<\infty$ for every $r>0$.
Fix $r>0$ and $\epsilon>0$. For $(t,x)\in \tE$ satisfying $|(t,x)|\leq r$, note that $|\frac1n b(t,x)+\frac{1}{\sn} \tilde{\sigma}(t,x)Z|\leq\frac1n\Mb(r)+\frac{1}{\sn}\Msigma(r)|Z|$, from which one may establish existence of $\gamma>0$ and $N_0>0$ (depending on $r$ and $\epsilon$) such that \begin{equation}\label{pr:An dans un ensemble ne dependant pas de x}
D_n(t,x,\epsilon)\subset\{|Z|>\sqrt{n}\,\gamma\}\,,\quad \text{for all }n\geq N_0\text{ and }|(t,x)|\leq r\,.
\end{equation}
If $|(t,x)|\leq r$ and $n\geq N_0$, we then have
$$
\tilde{\mathbb{E}}\left[|Z|^2\,\mathbb{I}_{D_n(t,x,\epsilon)}\right]\leq\tilde{\mathbb{E}}\left[|Z|^2\,
\mathbb{I}_{\{|Z|>\sqrt{n}\gamma\}}\right] ,
$$
and since the right-hand side does not depend on $(t,x)$, \eqref{eq:Z carre un An va a zero} easily follows from the dominated convergence theorem, and the fact that the components of $Z$ have finite second moments. The inclusion in~\eqref{pr:An dans un ensemble ne dependant pas de x} also yields the first inequality in
\begin{equation*}
n\,\tilde{\mathbb{P}}(D_n(t,x,\epsilon))\;\leq\;
\frac1{\gamma^2}(\sqrt{n}\gamma)^2\tilde{\mathbb{E}}\left[\mathbb{I}_{\{|Z|>\sqrt{n}\gamma\}}\right]\;\leq\;
\frac1{\gamma^2}\tilde{\mathbb{E}}\left[|Z|^2\,\mathbb{I}_{\{|Z|>\sqrt{n}\gamma\}}\right]\,,
\end{equation*}
for $|(t,x)|\leq r$ and $n\geq N_0$, and~\eqref{eq:n prob de An va a zero} is also immediate from the dominated convergence theorem.
\end{proof}

Note that~\eqref{eq:n prob de An va a zero} is equivalent to saying that, with $K_n$ defined in~\eqref{eq:transition function}, the supremum in equation~\eqref{eq:chaine de Markov ne saute pas beaucoup} tends to zero as $n$ goes to infinity, for all $\epsilon>0$ and $r>0$.

\begin{lem}\label{lemme:an et bn tout pres de a et b}
Assume Conditions~\ref{A:cond0}, \ref{A:cond1} and~\ref{A:cond2}. For the transition function $K_n$ defined in~\eqref{eq:transition function}, and $a_n$ and $b_n$ defined in~\eqref{eq:definition de an} and~\eqref{eq:definition de bn} respectively, the two sequences in~\eqref{eq:an et bn sont pres de a et b} tend to zero as $n$ goes to infinity, for all $r>0$.
\end{lem}
\begin{proof}
From~\eqref{eq:definition de bn} and $\tilde{\mathbb{E}}[Z]=0$ (recall Condition~\ref{A:cond2}), we can write
%
\begin{align*}
b_n(t,x)\;&=\;n\tilde{\mathbb{E}}\left[\left(\frac1n b(t,x)+\frac{1}{\sn}\tilde{\sigma}(t,x)Z\right)\mathbb{I}_{\,\overline{D_n(t,x,1)}}\right]\\
&=\;b(t,x)-b(t,x)\tilde{\mathbb{P}}(D_n(t,x,1))-\sqrt{n}\tilde{\sigma}(t,x)\tilde{\mathbb{E}}\left[Z\mathbb{I}_{D_n(t,x,1)}\right]\,.
\end{align*}
Here, $D_n$ is as in Lemma~\ref{lemme:resultats techniques sur An}.
Using the Cauchy-Schwarz inequality, and with $\Mb$ and $\Msigma$ as in the proof of Lemma~\ref{lemme:resultats techniques sur An}, we then see that, for $|(t,x)|\leq r$,
\begin{equation*}
|b_n(t,x)-b(t,x)|\leq\Mb(r)\tilde{\mathbb{P}}(D_n(t,x,1))+\Msigma(r)\left(\tilde{\mathbb{E}}|Z|^2\right)^\frac12
\,\left(n\tilde{\mathbb{P}}(D_n(t,x,1))\right)^\frac12\,.
\end{equation*}
The result follows by~\eqref{eq:n prob de An va a zero} in Lemma~\ref{lemme:resultats techniques sur An}. Similarly, note that
\begin{equation*}
a_n(t,x)=n\tilde{\mathbb{E}}\left[\left(\frac1n b(t,x)+\frac{1}{\sn}\tilde{\sigma}(t,x)Z\right)
\left(\frac1n b(t,x)+\frac{1}{\sn}\tilde{\sigma}(t,x)Z\right)^\top\mathbb{I}_{\,\overline{D_n(t,x,1)}}\right]\,.
\end{equation*}
Rearranging, and using the fact that $\tilde{\mathbb{E}}[ZZ^\top]=\Sigma$ and $a=\tilde{\sigma}\Sigma\tilde{\sigma}^\top$ (see Condition~\ref{A:cond2}), we get
\begin{align*}
a_n(t,x) & -a(t,x)\\
&= \frac1n b(t,x)b(t,x)^\top\tilde{\mathbb{P}}(\overline{D_n(t,x,1)}) - \tilde{\sigma}(t,x)\tilde{\mathbb{E}}[ZZ^\top\mathbb{I}_{\,D_n(t,x,1)}]\tilde{\sigma}(t,x)^\top\\
&+ \frac{1}{\sn} b(t,x)\tilde{\mathbb{E}}[Z^\top\mathbb{I}_{\,\overline{D_n(t,x,1)}}]\tilde{\sigma}(t,x)^\top + \frac{1}{\sn} \tilde{\sigma}(t,x)\tilde{\mathbb{E}}[Z\mathbb{I}_{\,\overline{D_n(t,x,1)}}]b(t,x)^\top.
\end{align*}
If $|(t,x)|\leq r$, we then have
$$
|a_n(t,x)-a(t,x)|\leq\frac1n\Mb(r)^2+\frac2{\sqrt{n}}\Mb(r)\Msigma(r)\tilde{\mathbb{E}}|Z|
+\Msigma(r)^2\tilde{\mathbb{E}}[|Z|^2\mathbb{I}_{D_n(t,x,1)}] ,
$$
and the result follows from~\eqref{eq:Z carre un An va a zero} in Lemma~\ref{lemme:resultats techniques sur An}.
\end{proof}

%
%
\bibliography{lrr}

\begin{thebibliography}{}

\bibitem[Alfonsi, 2005]{alfonsi2005}
Alfonsi, A. (2005).
\newblock On the discretization schemes for the {CIR} (and {B}essel squared)
  processes.
\newblock {\em Monte Carlo Methods Appl.}, 11(4):355--384.

\bibitem[Berkaoui et~al., 2008]{berkaouietal2008}
Berkaoui, A., Bossy, M., and Diop, A. (2008).
\newblock Euler scheme for {SDE}s with non-{L}ipschitz diffusion coefficient:
  strong convergence.
\newblock {\em ESAIM: Probability and Statistics}, 12:1--11.

\bibitem[Billingsley, 1968]{billingsley1968}
Billingsley, P. (1968).
\newblock {\em Convergence of probability measures}.
\newblock John Wiley \& Sons Inc., New York.

\bibitem[Bossy and Diop, 2007]{bossydiop2007}
Bossy, M. and Diop, A. (2007).
\newblock An efficient discretisation scheme for one dimensional {SDE}s with a
  diffusion coefficient function of the form $|x|^\alpha$, $\alpha \in
  [1/2,1)$.
\newblock Technical Report 5396, INRIA.

\bibitem[Broadie and Kaya, 2006]{broadiekaya2006}
Broadie, M. and Kaya, {\"O}. (2006).
\newblock Exact simulation of stochastic volatility and other affine jump
  diffusion processes.
\newblock {\em Oper. Res.}, 54(2):217--231.

\bibitem[Cox et~al., 1985]{coxetal1985}
Cox, J.~C., Ingersoll, Jr., J.~E., and Ross, S.~A. (1985).
\newblock A theory of the term structure of interest rates.
\newblock {\em Econometrica}, 53(2):385--407.

\bibitem[Deelstra and Delbaen, 1998]{deelstradelbaen1998}
Deelstra, G. and Delbaen, F. (1998).
\newblock Convergence of discretized stochastic (interest rate) processes with
  stochastic drift term.
\newblock {\em Appl. Stochastic Models Data Anal.}, 14(1):77--84.

\bibitem[Duffie et~al., 2000]{duffiepansingleton2000}
Duffie, D., Pan, J., and Singleton, K. (2000).
\newblock Transform analysis and asset pricing for affine jump-diffusions.
\newblock {\em Econometrica}, 68(6):1343--1376.

\bibitem[Ethier and Kurtz, 1986]{ethierkurtz1986}
Ethier, S.~N. and Kurtz, T.~G. (1986).
\newblock {\em Markov processes: characterization and convergence}.
\newblock John Wiley \& Sons Inc., New York.

\bibitem[Glasserman, 2004]{glasserman2004}
Glasserman, P. (2004).
\newblock {\em Monte {C}arlo methods in financial engineering}.
\newblock Springer-Verlag, New York.

\bibitem[Heston, 1993]{heston1993}
Heston, S.~I. (1993).
\newblock A closed-form solution for options with stochastic volatility with
  applications to bond and currency options.
\newblock {\em Review of Financial Studies}, 6:327--343.

\bibitem[Higham and Mao, 2005]{highammao2005}
Higham, D.~J. and Mao, X. (2005).
\newblock Convergence of {M}onte {C}arlo simulations involving the
  mean-reverting square root process.
\newblock {\em Journal of Computational Finance}, 8(3):35--62.

\bibitem[Kahl et~al., 2008]{kahletal2008}
Kahl, C., G\"{u}nther, M., and Rossberg, T. (2008).
\newblock Structure preserving stochastic integration schemes in interest rate
  derivative modeling.
\newblock {\em Applied Numerical Mathematics}, 58:284--295.

\bibitem[Kloeden and Platen, 1992]{kloedenplaten1992}
Kloeden, P.~E. and Platen, E. (1992).
\newblock {\em Numerical solution of stochastic differential equations},
  volume~23 of {\em Applications of Mathematics}.
\newblock Springer-Verlag, Berlin.

\bibitem[Lamberton and Lapeyre, 1991]{lambertonlapeyre1991}
Lamberton, D. and Lapeyre, B. (1991).
\newblock {\em Introduction au calcul stochastique appliqu\'e \`a la finance}.
\newblock Ellipses.

\bibitem[Lord et~al., 2010]{lordetal2010}
Lord, R., Koekkoek, R., and van Dijk, D. (2010).
\newblock A comparison of biased simulation schemes for stochastic volatility
  models.
\newblock {\em Quant. Finance}, 10(2):177--194.

\bibitem[Milstein et~al., 1998]{milsteinetal1998}
Milstein, G., Platen, E., and Schur, H. (1998).
\newblock Balanced implicit methods for stiff stochastic systems.
\newblock {\em SIAM Journal on Numerical Analysis}, 35(3):1010--1019.

\bibitem[Shreve, 2004a]{shreve2004a}
Shreve, S.~E. (2004a).
\newblock {\em Stochastic calculus for finance {I} - The binomial asset pricing
  model}.
\newblock Springer-Verlag.

\bibitem[Shreve, 2004b]{shreve2004b}
Shreve, S.~E. (2004b).
\newblock {\em Stochastic calculus for finance {II} - Continuous-time models}.
\newblock Springer-Verlag.

\bibitem[Stroock and Varadhan, 1979]{stroockvaradhan1979}
Stroock, D.~W. and Varadhan, S. R.~S. (1979).
\newblock {\em Multidimensional diffusion processes}.
\newblock Springer-Verlag, Berlin.

\end{thebibliography}
\bibliographystyle{apalike}

\end{document}